\documentclass[journal,12pt,onecolumn,draftclsnofoot,]{IEEEtran}
\usepackage[T1]{fontenc} 
\usepackage[english]{babel} 
\usepackage{amsmath,amsfonts,amsthm,amssymb,mathtools} 
\usepackage{algorithm}
\usepackage{cancel}
\usepackage{enumitem}
\usepackage[noend]{algpseudocode}
\usepackage{cite,graphicx,float}
\usepackage{subcaption}
\usepackage{tabularx}
\usepackage{color}
\usepackage{caption}
\usepackage{cite}
\usepackage{soul}
\usepackage[figurewithin=none]{caption}

\usepackage{fancyhdr} 
\newtheorem{theorem}{Theorem}
\newtheorem{lemma}{Lemma}

\newtheorem{prop}{Proposition}
\def\bm#1{\mathbf{#1}}

\newcommand{\squeezeup}{\vspace{-3mm}}
\newcommand{\at}[2][]{#1|_{#2}}

\title{	\Large Energy-Efficient Offloading in Delay-Constrained
Massive MIMO Enabled Edge Network Using Data Partitioning\thanks{The work in this paper is supported in part by the National Science Foundation under ECCS Grant 1808912. A part of this work is published in IEEE Globecom 2019~\cite{Malik2019}.}\\ 
}

\author{\normalsize{Rafia Malik and Mai Vu}\\
\small Department of Electrical and Computer Engineering, Tufts University, MA, USA\\
Email: rafia.malik@tufts.edu, mai.vu@tufts.edu} 

\date{\normalsize{May 21, 2018}} 

\begin{document}
\maketitle 
\begin{abstract}
We study a wireless edge-computing system which allows multiple users to simultaneously offload computation-intensive tasks to multiple massive-MIMO access points, each with a collocated multi-access edge computing (MEC) server. Massive-MIMO enables simultaneous uplink transmissions from all users, significantly shortening the data offloading time compared to sequential protocols, and makes the three phases of data offloading, computing, and downloading have comparable durations. Based on this three-phase structure, we formulate a novel problem to minimize a weighted sum of the energy consumption at both the users and the MEC server under a round-trip latency constraint, using a combination of data partitioning, transmit power control and CPU frequency scaling at both the user and server ends. We design a novel nested primal-dual algorithm using two different methods to solve this problem efficiently. Optimized solutions show that for larger requests, more data is offloaded to the MECs to reduce local computation time in order to meet the latency constraint, despite higher energy cost of wireless transmissions. Massive-MIMO channel estimation errors under pilot contamination also causes more data to be offloaded to the MECs. Compared to binary offloading, partial offloading with data partitioning is superior and leads to significant reduction in the overall energy consumption.
\end{abstract}


\section{Introduction}
Evolution of wireless communication networks towards denser deployments with large number of connected devices has led to an exponential growth in wireless traffic. Global mobile data traffic is expected to increase seven-fold from 2016 to 2021, of which, mobile video data by smartphones is the fastest-growing segment with a projected increase of 870\%~\cite{Cisco2017}. The trend towards \textit{smarter} smartphones has enabled new services such as Augmented Reality (AR), Virtual Reality (VR) and multi-user interaction, which further cause traffic surges particularly in localized live broadcast events such as a concert or a sports event. To cope with these demands, future generation networks including 5G and beyond are expected to handle multiple folds increase of data traffic at stringent latency requirements. One solution to this spike in latency-sensitive data demand is to bring data computation power closer to the devices at the multi-access edge computing (MEC) networks. MEC is a promising technology to provide cloud-computing capabilities within the Radio Access Network (RAN) in close proximity to mobile subscribers and eliminate the need to route traffic through the core network~\cite{MEC2014}. By moving the computing and storage features to the edge, MEC can offer a distributed and decentralized service environment characterized by proximity, low latency, and high rate access~\cite{Doppler2018}. 

Power hungry devices and computation intensive applications naturally lead to an escalated energy demand and therefore make energy efficiency a key parameter in the design of next generation networks. To this end, power management techniques in hardware are becoming popular. \textit{Dynamic Voltage and Frequency Scaling} (DVFS) is a common power saving technique which uses frequency scaling to reduce power consumption in a CMOS integrated circuit (e.g. the CPU~\cite{LeSueur2010}). A linear growth in the CPU frequency $f$ causes the dynamic power dissipation to increase cubically, leading to an energy consumption as $E_{dyn} \propto f^2$~\cite{Silva2018}. Therefore reducing the frequency leads to a dramatic reduction in energy consumption, which also holds true for modern processors with nanoscale features with non-negligible static power consumption~\cite{Silva2018}. DVFS has been traditionally used for personal computers and is now making its way to MEC servers and smart consumer devices including smartphones and tablets to conserve energy~\cite{Bezerra2013}. 

To handle the vast amount of services and computation requirements, MEC servers with high computation capacities also employ parallel computing via virtualization techniques to enable independent computation for each assigned user or task~\cite{Mao2017}. Network virtualization is a catalyst in supporting multi-tenancy and multiple services for edge computing architectures enabling efficient network operations and service provisioning. Virtualization technologies including network slicing, software defined networking (SDN), network function virtualization (NFV), virtual machines (VM), and containers are some of the key enablers of MEC networks~\cite{Taleb2017}. Using virtualization, the MEC server can optimally allocate processing frequency, or clock speed, per task or user such that each user can experience an independently orchestrated QoS, hence allowing the MEC to efficiently compute all users' tasks in parallel within the latency constraint.

To efficiently transfer data between user devices and MEC servers for computation, wireless base-stations/access-points (BS/AP - AP used in this paper interchangeably for both base station and access point) equipped with \textit{massive MIMO} technology can dramatically increase spectral efficiency by allowing the AP to simultaneously accommodate multiple co-channel users. The massive number of antennas at an AP can be used to create asymptotically orthogonal channels and deliver near interference-free signals for each user terminal~\cite{Wang2017}. For MEC architectures with co-located MEC server and AP~\cite{MEC2014}, the use of massive MIMO significantly reduces the wireless data transmission time especially in the uplink (data offloading) and hence has a drastic impact on the round-trip edge computation latency.

It is a realistic vision for future wireless networks to employ all aforementioned technologies: edge computing, massive MIMO transmission, network virtualization, and frequency scaling for power management. Such a network can be energy efficient in terms of both computation and communication while providing low-latency communication, and supporting highly-intensive computation tasks for their connected users. To achieve this vision, we will need to solve the intricate problem of optimal resource allocation, particularly in balancing local and MEC-offloaded computation, frequency allocation, energy consumption and time utilization.

\subsection*{Related Works}
Resource allocation in MEC networks has been an active area of research recently. Most existing works have considered energy minimization at only one side of the network, either the users~\cite{Scutari2015}\cite{Chae2016}\cite{You2017}, or the MEC-server when it is energy constrained, such as a UAV-MEC~\cite{Wang2018}\cite{Chu2018}. Common among prior works is the assumption of binary offloading, that is, each computation task is atomic and cannot be partitioned; hence these works examine a system-level problem with the perspective of choosing whether to offload a task to the MEC or to perform the computation locally. For example, several works minimize the system's computation overhead (energy and processing time)~\cite{Guo2018} and system-wide energy consumption~\cite{Leng2016}, while others consider a system utility function such as a weighted sum of the energy consumption and time delay in the entire system, considering users as mobile~\cite{Letaief2017} or generic connected devices~\cite{Sengul2017}. These works also assume that the user's clock frequency or computing capability is fixed, and therefore is not an optimizing variable.  Only recently,  partial offloading, where user tasks can partly be computed locally and partly offloaded to the MEC, has been considered for the problem of AP's energy minimization subject to users' latency requirement\cite{Wang2018}.

For multiuser MEC systems, the multiple access scheme affects edge computing latency significantly. Existing works typically employ Time Division Multiple Access (TDMA) for different users to sequentially offload information to the MEC in their designated time slots~\cite{Chae2016}\cite{You2017}\cite{Bi2018}\cite{Wang2018}. Under TDMA, the time spent for offloading computation tasks for all users in the uplink far exceeds the time for delivering results in the downlink, therefore, the latter is usually assumed to be negligible and is not factored into the round-trip latency. Such a latency constraint is important and has been considered in energy efficient computation for the users~\cite{Chae2016}\cite{Sengul2017}\cite{Bi2018} and for the MEC access points \cite{Wang2018}. Several works try to reduce the latency by assuming numerous channels available for offloading from users to the MECs, however, at the expense of consuming significantly more bandwidth~\cite{Fu2016}\cite{Liu2018}.

To solve for the different variations of resource allocation problems in MEC networks, algorithms with varying levels of complexity have been proposed. For example, centralized and distributed successive convex approximation (SCA) based algorithms are used in a static framework to reach local optimal solutions in a finite number of iterations~\cite{Scutari2015}.  A mixed integer non-linear problem is solved using bisection search and difference of convex optimization methods by decomposing the energy minimization problem into independent subproblems for individual users~\cite{Le2019}. A game-theoretic approach is used to find a near-optimal solution to the computational overhead (time and energy) minimization problem where convergence to the Nash equilibrium scales linearly with the number of computation tasks~\cite{Guo2018}. A distributed implementation of the offloading game achieves faster convergence compared to the centralized method at a small performance loss, with convergence speed scaling almost linearly with the number of users~\cite{Fu2016}.

\subsection*{Major Contributions}
In this work, we consider a multi-cell multi-user network scenario where access points equipped with massive MIMO antenna arrays and co-located MEC servers offer computation offloading. The novel feature of massive MIMO allows the users to offload their data to the MECs simultaneously, instead of using the sequential TDMA protocol, and hence significantly reduces the round-trip latency. We formulate a novel optimization problem to minimize the system's energy consumption, including both the users and the MEC, subject to a latency requirement. Our aim is to explore the benefit of computation offloading to meet a hard latency constraint while minimizing the energy consumption at both the user terminals and the MEC servers. The formulated problem befits edge network problems where computation offloading proves useful; for instance in AR/VR applications, a video surveillance system collecting data from multiple recording cameras, offloading data in real-time to the edge server for facial or object detection, or in real-time map rendering for autonomous vehicular applications, where computation offloading to the edge can be critical for real-time updates \cite{MEC2018}. The main contributions of this work can be summarized as follows.
\begin{enumerate}[leftmargin=*]
\item We show the immense benefits of massive MIMO in edge computing systems, which have not been explored earlier. Not only does the use of massive MIMO enable simultaneous (instead of sequential TDMA~\cite{Chae2016,You2017,Bi2018,Wang2018}) transmissions among multiple users, dramatically reducing offloading time and overall latency, it also reduces the transmit power at the AP for a given data rate and has a positive impact on the system energy consumption. Thus employing massive MIMO in an MEC system is beneficial for improving both latency and power consumption.

\item We propose a new formulation for MEC system-level energy minimization under massive MIMO employment. The formulation accounts for energy consumption at both the users and MEC ends, compared to current literature considering only one side \cite{Scutari2015,Chae2016,You2017,Wang2018,Chu2018}. Minimizing system level energy with delay and power constraints makes the problem not only richer but also more applicable in practice.

\item We design efficient, customized nested algorithms exploiting problem structure to solve for optimal resource allocation with potential for real-time implementation. The resource allocation is inclusive of data partitioning (partial offloading instead of binary offloading), time, power and computing frequency allocation, compared to the majority of current MEC literature which just optimizes for a part of these variables~\cite{Guo2018,Leng2016,Letaief2017,Sengul2017,Chae2016,You2017,Bi2018,Wang2018,Fu2016}. The algorithm with a nested structure, consisting of an outer latency-aware descent algorithm for data partitioning and an inner primal-dual algorithm for time, frequency and power allocation, is novel, efficient and guarantees solution optimality.

\end{enumerate}
\subsubsection*{Notation} $\boldsymbol{X}$ and $\boldsymbol{x}$ denote a matrix and vector respectively, $\nabla^2 f(x)$ denotes the Hessian matrix, and $\nabla^2 f(x)^{-1}$ denotes its inverse. For an arbitrary size matrix, $\boldsymbol{Y}$,  $\boldsymbol{Y}^\ast$ denotes the Hermitian transpose, and $\textbf{diag}(y_1,...,y_N)$ denotes an $N\times N$ diagonal matrix with diagonal elements $y_1,...,y_N$. $\boldsymbol{I}$ denotes an identity matrix, and $\boldsymbol{0, 1}$ denote an all zeros and all ones vector respectively. The standard circularly symmetric complex Gaussian distribution is denoted by $\mathcal{CN}(\boldsymbol{0}, \boldsymbol{I})$, with mean $\boldsymbol{0}$ and covariance matrix $\boldsymbol{I}$. $\mathbb{C}^{k \times l}$ and $\mathbb{R}^{k \times l}$ denote the space of $k \times l$ matrices with complex and real entries, respectively.

\section{System Model}
We consider a system with $L \geq 1$ Access Points (APs), each equipped with a massive-MIMO array of $N$ antennas deployed over a target area, for instance in a sports stadium, a town fair or a crowded exhibit or mall. We consider a deployment scenario where a Multi-access Edge Computing (MEC) server is collocated with each AP\cite{MEC2014}. Single-antenna users offload data in the uplink to the MEC-APs and receive computed results in the downlink. Each AP serves an area denoted as a cell, which contains $K$ users, making the system a multi-cell network.
\setlength{\belowcaptionskip}{-20pt}
\begin{figure}[t]
\centering
\includegraphics[scale = 0.57]{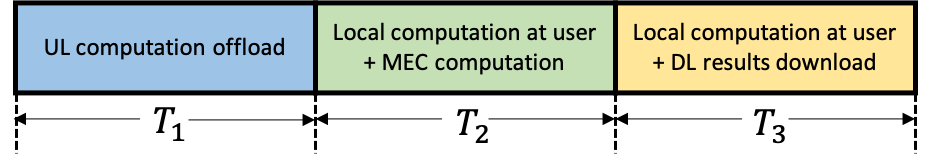}
\caption{Operation phases at each MEC server }
\label{fig:phases}
\end{figure}

Each MEC server schedules data offloading for the users that it serves. We consider the type of applications typically composed of multiple procedures, for example computation components in an AR application, and hence can support partial offloading or program partitioning~\cite{Mao2017}. For such applications allowing partial offloading, for the $i^{th}$ user, the $u_i$ computation bits can be partitioned into $q_i$ bits to be computed locally and $s_i$ bits to be offloaded to the MEC server. We consider the data-partition model where the computation task is bit-wise independent, and also assume that such partition does not incur additional computation bits, that is, $u_i = q_i + s_i$~\cite{Wang2018}. The data-partition task model is applicable for semantic image segmentation in map-rendering applications~\cite{Buslaev2018}, or in modern technologies employed in AR/VR applications, such as multi-user encoding~\cite{Hou2017}, multicasting and tiling~\cite{Dai2019}\cite{He2018}, among others, in which edge computing servers can transcode and stitch the data into a seamless real-time stream~\cite{Intel2017}. For a given number of computation requests, we examine the problem of resource allocation for completing the targeted tasks within the latency constraint in the most energy efficient manner.

Given a total latency constraint denoted as $T_d$, the time for data offloading, computation (at both users and MEC), and the delivery of computed results to the users should not exceed $T_d$. The system's operation can hence be divided into three phases as shown in Figure~\ref{fig:phases}; (i) computation offloading from users to the MEC in the uplink, (ii) computation at the MEC server and locally at the user, and (iii) transmission of computed results from the MEC to users in the downlink. Note that local computation at the user can span both phases two and three. Simultaneous transmission from multiple users in the uplink through the use of massive MIMO significantly shortens the offloading time (Phase I), making downloading time (Phase III) no longer negligible as was with TDMA offloading~\cite{Chae2016}\cite{You2017}\cite{Bi2018}\cite{Wang2018}. A subsequent benefit of this non-negligible downloading time is that users can now perform local computation through the time in both phases II and III.

Energy at both the MEC and the user terminal is consumed for two tasks; 1) for data computation which depends on the CPU frequency, and 2) for transmitting the data for computation offloading or delivering results over the uplink or downlink channel respectively. CPU frequency is an important parameter which affects both time and energy consumptions. While a higher CPU frequency implies lesser computation time, it also increases energy consumption~\cite{Chae2016}. Therefore optimizing CPU frequency using DVFS can achieve energy efficient computation~\cite{LeSueur2010}. For MEC servers with high computation capabilities, we assume virtualization in our system model such that the MEC can optimally allocate processing frequency per user. In this way, the MEC can efficiently compute all users' tasks in parallel within the constrained latency.

\setlength{\belowcaptionskip}{-20pt}
\begin{figure}[t]
        \begin{minipage}{0.55\textwidth}
                \centering
                \includegraphics[scale = 0.55]{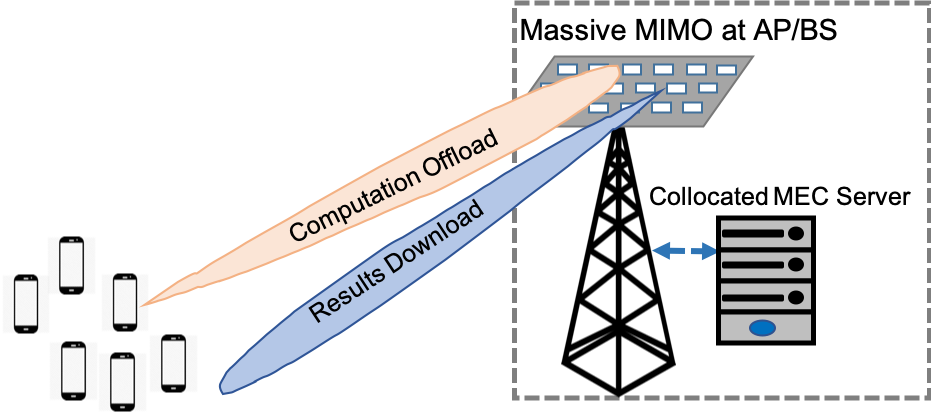}
                \caption{Beamforming using massive MIMO antenna array at AP/BS}
                \label{beamform}
        \end{minipage}
       \hfill
        \begin{minipage}{0.45\textwidth}
                \centering
                \includegraphics[scale=0.4]{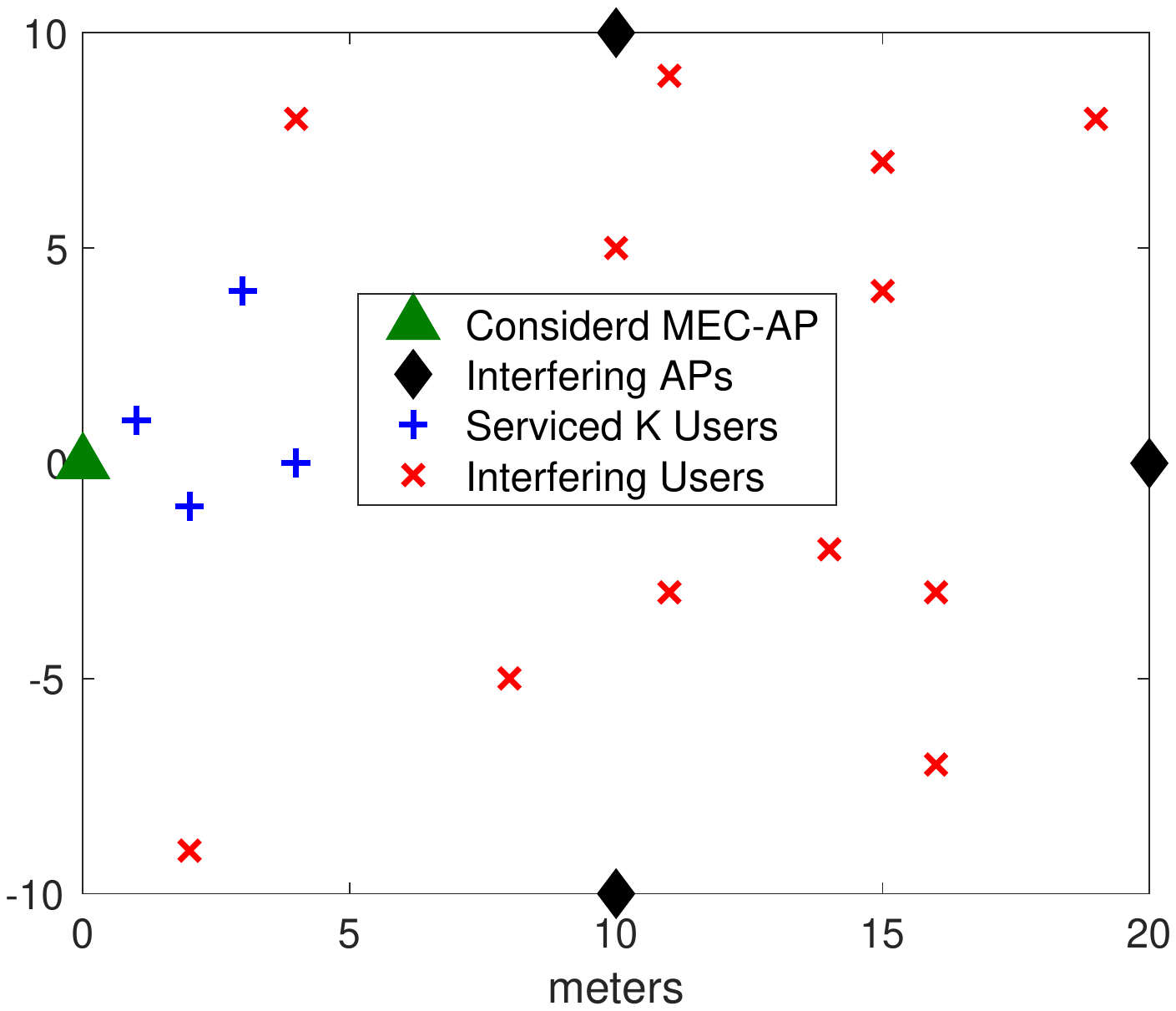}
                \caption{Network Layout}
                \label{system_model}
        \end{minipage}
\end{figure}
In the considered multi-cell environment, we discuss and set up the problem for a typical cell denoted as the \textit{home cell}. Assuming a frequency reuse factor of 1 (as typical in LTE networks~\cite{Ali2011}), other cells which use the same pilots as the home cell are called \textit{contaminating cells}. The effect of the multi-cell environment is taken into account implicitly via inter-cell interference and massive MIMO pilot contamination. Each user has a computing requirement for a certain amount of data, which can be divided into a part for local computation and a part for offloading to the MEC. The partitioned amount of data $\boldsymbol{s}$ to be offloaded to the MEC is a key design variable which spans all the phases of the system's operation. Next we discuss the time and energy consumption while highlighting the design variables for resource allocation specific to each phase of operation, in addition to $\boldsymbol{s}$ which is a design variable across all three phases.
\subsection{Phase I: Computation Offloading in Uplink}
\subsubsection{Data Transmission}
In a given time slot, $K$ user terminals concurrently offload data to the $N$-antenna AP over the uplink channel in the same time-frequency resource. We define $\beta_i \triangleq S_{\sigma}d_{i}^{-\gamma}$ as the large scale fading between the $i^{\text{th}}$ user and the AP, assuming it to be the same for all AP antennas (independent of $N$), where $S_{\sigma}$ denotes log-normal shadowing with standard deviation $\sigma \text{ dB}$, $d_{i}$ is the distance from the $i^{th}$ user to the AP, and $\gamma$ is the path loss exponent.

With a sufficient number of antennas, the channel hardens such that the effective channel gains become nearly deterministic~\cite{Larsson2017}. This channel hardening effect has been observed experimentally for a massive MIMO system built specifically for MEC application with 128 antennas~\cite{Sara2018}. We consider the operating regime with $N \gg K$ typical for a TDD massive MIMO system, in which the throughput becomes independent of the small-scale fading with channel hardening~\cite{Marzetta2016}. This throughput depends on the type of detector employed at the massive MIMO terminal. We consider maximum ratio combining, for which the uplink achievable transmission rate for the $i^{th}$ user in the $l^{th}$ cell, $r_{u,i}$, is given as~\cite{Marzetta2016}
\begin{equation}\label{rate_ul}
r_{u,i} = \nu \log_2 \left ( 1 + \frac{\text{SINR}_{li}^{ul}}{\Gamma_{1}} \right ), \ \text{SINR}_{li}^{ul} = \frac{N \gamma_{li}^l p_{li}}{\sigma_{1,li}^2}
\end{equation}
where $\Gamma_{1} \geq 1$ is a constant accounting for the capacity gap due to practical coding and modulation schemes, $p_{li}$ is the transmit power of the $i^{\text{th}}$ user in the $l^{\text{th}}$ cell. Here the constant $\nu = \frac{\tau_c - \tau_u}{\tau_c}$ accounts for the effective loss of samples due to the transmission of pilot symbols in each coherence interval for channel estimation at the AP, where $\tau_c$ is the length of the coherence interval and $\tau_u$ is the duration of pilot transmission. We follow standard practice of using the critical number of pilot symbols equaling the number of users: $\tau_u = K$~\cite{Ngo2014}. The term $\sigma_{1,li}^2$ is the interference and noise power including the effect of pilot contamination and intercell interference as 
\begin{equation}\label{INul}
\sigma_{1,li}^2 = \sigma_r^2 + \sum_{q \in \mathcal{P}_l} \sum_{i=1}^{K} \beta^l_{qi} p_{qi} + \sum_{q \notin \mathcal{P}_l} \sum_{i=1}^{K} \beta^l_{qi} p_{qi} + N \sum_{q \in \mathcal{P}_l \backslash l} \gamma_{qi}^l p_{qi}
\end{equation}
where $\sigma_r^2$ is the receiver noise variance, the second term represents interference from contaminating cells, the third term is inter-cell interference, and the last term is interference due to the mean-square channel estimates from contaminating cells excluding the home cell and is also called the coherent interference~\cite{Marzetta2016}.

\subsubsection{Energy and Time Consumption}
An offloading overhead is incurred for transmitting the offloaded bits over the uplink channel to the MEC server. The energy consumed for offloading the $i^{th}$ user's data is given by $E_{OFF,i} = p_{li} t_{u,i}$, where $p_{li}$ is the transmit power and $t_{u,i}$ is the transmission time for the $i^{th}$ user. Let $B$ denote the channel bandwidth, then $t_{u,i} = \frac{s_i}{B r_{u,i}}$. All users offload their computation bits simultaneously, and the total energy and time consumptions for Phase I can then be written as
\begin{equation}\label{E_ul}
E_{OFF} = \sum_{i=1}^{K} \frac{p_{li} s_i} {B r_{u,i}}, \ T_1 = \max_{i \in [1,K]} t_{u,i}.
\end{equation}
In this phase, the offloading time $\boldsymbol{t_u} = [t_{u,1}...t_{u,K}] \in \mathbb{R}^{K \times 1}$ is a design variable to be optimized which also implicitly affects the transmit power $\boldsymbol{p} = [p_{l1}...p_{lK}] \in \mathbb{R}^{K \times 1}$ as will be shown later.

\subsection{Phase II: Computation at MEC Server and User Terminals}
\subsubsection{Local Computation at User Terminal}
Using DVFS architecture, the energy consumption and the processing time for local computation at the $i^{th}$ user is given as~\cite{Mao2017} 

\begin{align}\label{t_Li}
&E_{LC} = \sum_{i=1}^{K} \kappa_i c_i (u_i - s_i) f_{u,i}^2, \ \ t_{L,i} = \frac{c_i (u_i - s_i)}{f_{u,i}}
\end{align}
where $\kappa_i$ is the effective switched capacitance, $f_{u,i}$ denotes the average CPU frequency, $c_i$ denotes the CPU cycle information, that is, the number of CPU cycles required for computing one input bit, and $q_i = u_i - s_i$ is the total number of bits required to be locally computed at $i^{th}$ user respectively. Frequency scaling can be performed per CPU cycle, however, this causes large optimization overhead. We therefore consider average CPU frequency optimization. The users' local computation time can also extend to Phase III while the MEC is sending computed results back to users. This fact is considered later in the problem formulation (constraint d).

\subsubsection{Computation at the MEC server}
Assuming that the MEC servers have high computation capacities and utilize parallel computing via virtualization for independent computation per user, the energy consumed for computing offloaded bits of all users is given as
\begin{equation}\label{E_OC}
E_{OC} = \sum_{i = 1}^{K} \kappa_m f_{mi}^2 d_m s_i
\end{equation}
where $s_i$ is the number of bits offloaded by the $i^{th}$ user to the MEC, $d_m$ is the number of CPU cycles required to compute one bit at the MEC, the CPU frequency $f_{mi}$ is the computational rate assigned to the $i^{th}$ user's task by the MEC, and $\kappa_m$ is a hardware dependent constant of the MEC server. The computation time for processing the offloaded bits of $K$ users via parallel processing is given as $T_2$ below where $t_{M,i}$ is the time for computing $i^{th}$ user's offloaded task
\begin{equation}\label{tMEC}
T_2 = \max\{t_{M,i}\}, \ t_{M,i} = \frac{d_m s_i}{f_{mi}} \ \forall i \in [1, K].
\end{equation}
In this phase, the allocated CPU frequencies at the users, $\boldsymbol{f_u} = [f_{u1}...f_{uK}] \in \mathbb{R}^{K \times 1}$, and at the MEC, $\boldsymbol{f_m} = [f_{m1}...f_{mK}] \in \mathbb{R}^{K \times 1}$, are design variables. Note that the time for local computation $\boldsymbol{t_L} = [t_{L1}...t_{LK}] \in \mathbb{R}^{K \times 1}$ and offloaded computation $\boldsymbol{t_M}= [t_{M1}...t_{MK}] \in \mathbb{R}^{K \times 1}$ are directly affected by $\boldsymbol{f}$ and $\boldsymbol{f_m}$ as given in (\ref{t_Li}) and (\ref{tMEC}).

\subsection{Phase III: Delivering Computed Results in Downlink}
For downlink transmission we consider Time Division Duplex (TDD) operation such that the channel estimates in the uplink can be used for the downlink via reciprocity. With a sufficient number of antennas at the AP, not only do the effects of small scale fading and frequency dependence disappear due to channel hardening, but also channel estimation at the terminals, and the associated transmission of downlink pilots becomes unnecessary~\cite{Marzetta2016}\cite{Larsson2017}.

For the $i^{th}$ user in the $l^{th}$ cell (home cell), the downlink transmission rate with maximum ratio linear precoding at the MEC-AP is given as~\cite{Marzetta2016}
\vspace{-2mm}
\begin{equation}\label{rate_dl}
r_{d,i} = \log_2 \left ( 1 + \frac{\text{SINR}_{li}^{dl}}{\Gamma_{2}} \right ), \ \text{SINR}_{li}^{dl} = \frac{N P \gamma_{li}^l \eta_{li}}{\sigma_{2,li}^2}
\end{equation}
where $\Gamma_{2} \geq 1$ is the capacity gap similar to (\ref{rate_ul}), interference and noise power term is
\begin{equation}\label{INdl}
\sigma_{2,li}^2 = \sigma_i^2 + P \sum_{q \in \mathcal{P}_l} \sum_{i=1}^{K} \beta^l_{qi} \eta_{qi} + P \sum_{q \notin \mathcal{P}_l} \sum_{i=1}^{K} \beta^l_{qi} \eta_{qi}  + N P \sum_{q \in \mathcal{P} \backslash l} \gamma_{qi}^q \eta_{qi}
\end{equation}
where $\sigma_i^2$ is the noise at the $i^{\text{th}}$ user terminal in the $l^{\text{th}}$ cell, $\{\eta_{li}\} \in [0,1]$ are the power coefficients satisfying $\sum_{i=1}^{K} \eta_{li} \leq 1$ for all $l$, and $P$ is the AP's average transmit power. Similar to the uplink transmission, the second term in (\ref{INdl}) is pilot contamination, the third term is inter-cell interference which manifests as uncorrelated noise in the home cell, and the last term is coherent interference resulting from mean-square channel estimation errors. Since there is no pilot transmission in this phase, the effective downlink transmission rate is equal to the data rate.

\subsubsection{Energy and Time Consumption}
The transmission time for delivering the $i^{th}$ user's computation results can be written in terms of the downlink rate in (\ref{rate_dl}) as $t_{d,i} = \frac{\tilde{s}_i}{B r_{d,i}}$. Here $\tilde{s}_i$ denotes the number of information bits generated after processing $s_i$ offloaded bits of the $i^{th}$ user, and is assumed to be proportional to $s_i$, that is $\tilde{s}_i = \mu s_i$.  The proportionality ratio $\mu$ between the offloaded data and the computed results adds an application-centric flexibility to our system model in terms of the data size in downlink. For instance, for applications such as face recognition in a scenario where data from multiple video recording cameras is offloaded to the edge server for analysis, the computed results would be smaller in size than the offloaded data, in which case $\mu < 1$ can be chosen \cite{Chen2015}. On the other hand, for video-rendering applications such as those delivering $360^\circ$ videos in mobile networks, the ratio between the Field Of View (FOV) and the source video can be such that in order to provide a 4K video at the user device, the source video must be delivered over the network at a 16K resolution, which leads to $\mu \gg 1$ \cite{Mangiante2017}. The AP simultaneously transmits computed results for all users, and the total energy and time overhead for Phase III are then given as
\begin{equation}\label{E_dl}
E_{DL} = \sum_{i=1}^{K} \frac{P \eta_{li} \mu s_i}{B r_{d,i}}, \ T_3 = \max_{i \in [1,K]} t_{d,i}.
\end{equation}
where the total energy consumption is the sum of energy for data transmissions to all users, and the time consumption is the maximum among all users because of simultaneous transmissions in the downlink. In this phase, the downloading time $\boldsymbol{t_d}= [t_{d,1}...t_{d,K}] \in \mathbb{R}^{K \times 1}$ is a design variable for optimal resource allocation and also implicitly affects the power allocation $\boldsymbol{\eta} = [\eta_{l1}...\eta_{lK}] \in \mathbb{R}^{K \times 1}$ in the downlink at the AP.

\section{Energy Optimization Formulation}
Considering a multi-cell multi-MEC network, we formulate a novel optimization problem to minimize the weighted energy at both the MEC and users in the home cell, taking into account the effect from other cells via intercell interference. We then analyze the problem formulation to prepare for algorithm design in the next section.
\subsection{Weighted Energy Minimization Problem Formulation}
We formulate an edge computing problem which explicitly accounts for physical layer parameters including available transmit powers from each user and the MEC, associated massive MIMO data rates with realistic pilot contamination and interference. The problem jointly optimizes for the amount of partial data offloaded from each user $\boldsymbol{s}$, the CPU frequency for local computation at each user $\boldsymbol{f_u}$, the CPU frequency at the MEC allocated to each user's data computation $\boldsymbol{f_m}$, time allocation for uplink and downlink transmission $\boldsymbol{t_u, t_d}$, and the time duration for each phase, $T_1, T_2 $ and $T_3$, within a total latency requirement.
 
The total energy consumption by all users, based on equations (\ref{t_Li}) and (\ref{E_ul}), can be written as
\begin{equation}\label{E_u}
E_u =  \sum_{i=1}^{K} \left[\frac{t_{u,i}(2^{\frac{s_i}{\nu t_{u,i}B}} - 1)\Gamma_{1}\sigma_{1,i}^2}{N \gamma_i} + \kappa_i c_i (u_i - s_i) f_{u,i}^2 \right]
\end{equation}
Similarly, the total energy consumption at the MEC server, based on equations (\ref{E_OC}) and (\ref{E_dl}), is
\begin{equation}\label{E_m}
E_m = \sum_{i=1}^{K} \left[\frac{t_{d,i}(2^{\frac{\mu s_i}{t_{d,i}B}} - 1)\Gamma_{2}\sigma_{2,i}^2}{N \gamma_{i}}  + \kappa_m d_m f_{mi}^2 s_i \right]
\end{equation}
In these expressions, using (\ref{rate_ul}) and (\ref{rate_dl}), and by definition of the uplink and downlink transmission rates as $r_{u,i} = \frac{s_i}{\nu t_{u,i} B}$ and $r_{d,i} = \frac{\mu s_i}{t_{d,i} B}$ respectively, we have implicitly replaced the power allocation variables for per-user uplink transmission power ($p_{li}$) and per-user downlink power ($\eta_{li}$) as functions of the time allocation and data partitioning as follows
\begin{align}\label{poweralloc}
p_{li} = \frac{(2^{\frac{s_i}{\nu t_{u,i} B}} - 1)\Gamma_{1}\sigma_{1,i}^2}{N \gamma_{i}}, \  \ \eta_{li} = \frac{(2^{ \frac{\mu s_i}{t_{d,i} B}} - 1)\Gamma_{2}\sigma_{2,i}^2}{P N \gamma_{i}}
\end{align}Based on these expressions, a weighted system energy minimization can be formulated as
\setcounter{equation}{13}
\begin{align*}\label{Pnew}
\bm{(P)} \ \ &\min_{\boldsymbol{t, f, s}} \ \ \ E_{\text{total}} = (1 - w) E_{u} + w E_{m}  \tag{\theequation}&\\
&\text{ s.t. }  \ \ \  \text{Eqs. } (\ref{E_u})-(\ref{E_m}) \tag{a-b}&\\
&\hphantom{\text{ s.t.}}  \ \ \  \sum_{j=1}^{3} \left ( T_j\right ) = T_d, \ \ \frac{c_i (u_i - s_i)}{f_{u,i}} + t_{u,i} - T_d \leq 0 \tag{c-d} , \forall i \in [1,K]&\\
&\hphantom{\text{ s.t.}}  \ \ \  t_{u,i} - T_1 \leq 0, \ \ \  \frac{d_m s_i}{f_{mi}} - T_2 \leq 0, \ \ \ t_{d,i} - T_3 \leq 0, \forall i \in [1,K] \tag{e-g}&\\
&\hphantom{\text{ s.t.}}  \ \ \  \sum_{i=1}^{K_u} f_{mi} - f_{m,\max} \leq 0 \tag{h}&
\end{align*}
Here $E_{\text{total}}$ is weighted sum of energy consumption at all users ($E_{u}$) and the MEC ($E_{m}$), with $1 - w$ and $w$ as the respective weights. The optimizing variables of this problems are $\boldsymbol{t} = [t_{u,1}...t_{u,K}, t_{d,1}...t_{d,K}, T_1, T_2, T_3]$, $\boldsymbol{f} = [f_{u1}... f_{uK}, f_{m1}...f_{mK}]$ and $\boldsymbol{s} = [s_1...s_K]$. Implicit constraints not mentioned are  $f_{i,\min} \leq f_{u,i} \leq f_{i,\max}$ and $f_{m,\min} \leq f_{mi} \leq f_{m,\max} \ \forall i \in [1, K]$. Given parameters of the problems are $T_d$ as the total latency constraint, $P$ as the AP's transmit power, $B$ as the channel bandwidth, $\Gamma_1$, $\Gamma_2$ as the uplink and downlink capacity gaps, $(\kappa_i, c_i)$ and $(\kappa_m, d_m)$ as the switched capacitance and CPU cycle information at the users and the MEC respectively.

Constraint (c) shows that both the time consumed for all three phases at the MEC, and the time consumed for offloading and local computation at each user should not exceed $T_d$. Constraints (e-g) show that the time consumed separately for offloading, computation of users' tasks at the MEC, and downloading time for each user's results must be less than the maximum allowable time, $\{ T_1, T_2, T_3\} $, for that phase as given in \{(\ref{E_ul}),(\ref{tMEC}), (\ref{E_dl})\} respectively. Constraint (h) denotes the maximum CPU frequency at the MEC, which implies that with virtualization, the sum of frequencies allocated to all users' tasks should not exceed the MEC processor's capability. 

\subsection{Problem Analysis and Decomposition}
Problem (P) is a complicated multi-variable non-linear optimization which is also non-convex. This is due to constraint (\ref{Pnew}b) in which the term $s_i f_{mi}^2$ is neither convex nor concave since its Hessian is indefinite with one positive and one negative eigenvalue, making this constraint and consequently problem (P) non-convex. Next, we present analysis results which can be used to decompose this problem into two simpler convex sub-problems.

\begin{lemma}\label{lemma1}
The objective function $f_0$ of the problem (P) is convex as a function of $s_i$. Furthermore, if the system parameters satisfy the following condition which signifies a typical network setting where wireless transmission energy is non-negligible compared to computation energy:
\setcounter{equation}{14}
\begin{align*}\label{gradf0}
&\frac{(1 - w) 2^{\frac{s_i}{\nu t_{u,i}B}} \ln2 \Gamma_1 \sigma_{1,i}^2}{\nu B N \gamma_i} + \frac{w \mu 2^{\frac{\mu s_i}{t_{d,i}B}} \ln2 \Gamma_2 \sigma_{2,i}^2}{B N \gamma_i} + w \kappa_m d_m f_{m,i}^2 - (1 - w) \kappa_i c_i f_{u,i}^2 \Bigg \vert_{s_i \to 0} \geq 0 \tag{\theequation}
\end{align*}
then the total energy in problem (P) is an increasing function of each $s_i$. If condition (\ref{gradf0}) does not hold, then there exists a unique value of $s_i$ that minimizes the objective function $f_0$ obtained by solving $\nabla f_0(s_i) = 0$.
\end{lemma}
\begin{proof}
Let $f_0(.)$ be the objective function in (\ref{Pnew}). The second-order derivative for the objective function with respect to $s_i$ is
\begin{align*} 
\nabla^2_{s_i} f_0(s_i) &= \frac{(1 - w) 2^{\frac{s_i}{\nu t_{u,i}B}} (\ln2)^2 \Gamma_1 \sigma_{1,i}^2}{\nu^2 B^2 N \gamma_i t_{u,i}} + \frac{w \mu^2 2^{\frac{\mu s_i}{t_{d,i}B}} (\ln2)^2 \Gamma_2 \sigma_{2,i}^2}{B^2 N \gamma_i t_{d,i}}
\end{align*}
which is positive for all considered ranges of problem parameters. Thus, $f_0$ is a convex function of $s_i$. The expression in Lemma~\ref{lemma1} is the gradient of $f_0(\cdot)$ with respect to $s_i$ evaluated at $s_i = 0$, Since the gradient expression is an increasing function of $s_i$, if the gradient at $s_i$ approaching $0$ is non-negative ($\nabla_{s_i} f_0 (s_i) \at {s_i \to 0} \geq 0$) then the gradient is non-negative for all $s_i \geq 0$ and the Lemma follows directly.
\end{proof}
\textbf{Discussion: }Since the objective function is convex in $s_i$, there exists an optimal point, $s_i^\star \ \forall i \in [1,K]$, which minimizes $E_{\text{total}}$. We write the gradient expression in (\ref{gradf0}) as
\begin{equation}\label{gradbreak}
\nabla_{s_i} f_0 (s_i) = \nabla_{s_i} E_{OFF,i} + \nabla_{s_i} E_{DL,i} + \nabla_{s_i} E_{OC,i} + \nabla_{s_i} E_{LC,i}
\end{equation}
and define two cases for finding the optimal $s_i^\star$ as follows.
\subsubsection{Case I: Condition in (\ref{gradf0}) holds}
Here the total energy consumption is an increasing function of $s_i$, thus the optimal  $s_i^\star$ is as small as possible subject to the constraints of the problem (P). The first two terms in (\ref{gradbreak}) denote the rate of change, with respect to $s_i$, of energy consumption in wireless transmission in uplink and downlink, respectively. The last two terms represent the rate of change, with respect to $s_i$, of energy consumption in computation at the MEC and locally at the user respectively. Note that $E_{OFF,i}, E_{DL,i}$ and $E_{OC,i}$ are all increasing functions of the offloaded bits $s_i$ and positive, while $E_{LC,i}$ is negative. A positive overall gradient therefore implies that $\nabla_{s_i} E_{OFF,i} + \nabla_{s_i} E_{DL,i} + \nabla_{s_i} E_{OC,i} > \nabla_{s_i} E_{LC,i}$, which typically holds true for practical scenarios of typical network settings, with multiple APs and users located in a reasonable size target area, due to the dominant energy consumptions for wireless transmissions and MEC computation over that of local computation.
\subsubsection{Case II: Condition in (\ref{gradf0}) does not hold}
This case implies there exists $s_i^\star > 0$ which minimizes the objective function for the weighted sum energy. This case only holds if $w \to 0$, such that the problem is reduced to that of user energy minimization, since if $w \neq 0$, the gradient would be positive even for negligible transmission loss, because $f_{m,i} > f_{u,i}$ making $\nabla E_{OC,i} > \nabla E_{LC,i}$. This scenario only arises in non-typical settings, for example a single AP serving a single user at the close distance of $3$m, the minimum required separation between a Femtocell-AP and a user terminal (UT)~\cite{Clerckx2013}. For this case, condition (\ref{gradf0}) is reversed under $w = 0$, as the two middle terms in (\ref{gradbreak}) vanish, and the negligible transmission loss makes $\nabla_{s_i} E_{LC,i} > \nabla_{s_i} E_{OFF,i}$. Thus to conserve the user's energy, the optimal solution here is to offload all its data to the MEC thanks to the proximity to the MEC-AP. For most networks, however, if the MEC energy is also taken into account $(w > 0)$ or at a larger UT-AP distance then it may never be energy-optimal to offload all data to the MEC.

For the rest of the paper, we assume a typical network setting where Lemma \ref{lemma1} always holds true. Since the system energy is then increasing with the amount of offloaded data, it is of interest for the system to keep the offloaded data to minimum, only offloading when local computation violates power or latency constraints. Note that for non-typical networks, the solution approach presented in the next section would still hold except that we need to slightly modify the outer algorithm in Section \ref{outer_opt} to take into account the solutions of $\nabla f_0(s_i) = 0$ while keeping the latency constraint in check. Because of space constraint, we will focus on the typical network case only.

If the amount of offloaded data is given, then all we need to do is solve problem (P) for the remaining variables. The following lemma provides a theoretical basis for doing that.
\begin{lemma}\label{lemma2}
For a given set of offloaded data $s_i$, the problem (P) is convex in the remaining optimizing variables $\boldsymbol{t,f}$.
\end{lemma}
\begin{proof}
Proof follows by examining each constraint and showing that with fixed $s_i$, it is a convex function. Details in Appendix A.
\end{proof}
\section{Optimal Solution and Algorithms}\label{SolAlgo}
While problem (P) is not convex in all the optimizing variables, Lemma \ref{lemma2} shows that by fixing the offloaded bits $\boldsymbol{s}$, the problem is convex in all the remaining optimizing variables with a convex objective function and a convex feasible set. We can therefore divide problem (P) into two sub-problems: problem (P1) solves for the optimal balance between offloaded bits and those retained at the users, while (P2) solves for the remaining optimizing variables for a fixed number of offloaded bits $\boldsymbol{s}$. Since (P2) is convex, any algorithm which solves a convex problem can be applied, however, standard convex-solvers are often inefficient due to their inability to exploit the specific problem structure. We therefore analyze the problem in detail in both the primal and dual spaces to provide insight into the problem structure and propose a customized nested algorithm to solve problem (P) efficiently. 

The nested algorithm structure for solving (P) is shown in Figure~\ref{nested_algo} and the proposed algorithm works as follows. We first initialize the offloaded bits $\boldsymbol{s}$ and also initialize the dual variables. At the current value of $s$, the inner algorithm is executed, for which we use a primal-dual approach employing a subgradient method to solve sub-problem (P2). At each iteration of the inner algorithm, the current values of the dual variables are used to calculate the primal variables as stated in Theorems \ref{theorem1} and \ref{theorem2} below and also to determine the value of the dual function, then the dual variables are updated according to their respective subgradients. This process is repeated until the stopping criterion for the dual problem is satisfied, at which point the inner algorithm returns the control to the outer algorithm. Based on the newly updated primal solution for (P2) from the inner algorithm, we proceed to updating $\boldsymbol{s}$ for the next iteration of the outer algorithm, using a descent method while keeping in check the latency constraint. These steps for outer-inner optimization are repeated until a minimum point for the weighted total energy consumption is reached where all the constraints in the original problem (P) are satisfied.

The proposed nested algorithm is designed in a way to support its possible implementation in a real-time network scenario. This would imply that the algorithm is adaptable to changes in the network. For example, if a new user joins the network, or a current user leaves the network, the input to the outer optimization algorithm changes, and it forwards the updated network parameters to the inner optimization algorithm when making the function call. Therefore, a change in the network directly warrants an updated optimal solution, with no need for any change in the underlying two algorithms for solving (P1) and (P2). We now proceed to deriving the optimal solution for these two sub-problems, or equivalently for the problem (P).

\subsection{Latency-Aware Descent Algorithm for Outer Optimization}\label{outer_opt}
\begin{algorithm}[t]
\caption{Solution for Energy Minimization Problem (P)} 
\textbf{Given:} Distances $d_{i} \forall i$. Channel $\boldsymbol{H = G^{T}}$. Precision, $\epsilon_1, \epsilon_2$, Data amount $u_i$, Latency $T_d$\\
\textbf{Initialize:} Primal variable $s_i$.\\
\textbf{Begin Latency-Aware Descent Algorithm for (P1)}\\
\textbf{Given} a starting point $\boldsymbol{s}$\\
\textbf{Repeat}
\begin{enumerate}
\item Compute $\Delta \boldsymbol{s}$
\item Call the inner optimization algorithm, Algorithm 2
\item \textit{Line search}. Choose step size $t_i$ for each user via backtracking line search.
\item \textit{Update}. $s_i := s_i + t_i\Delta s_i$.
\end{enumerate}
\textbf{Until} stopping criterion is satisfied with $\epsilon_1$ or latency constraint $T_d$ is met.\\
\textbf{End Latency-Aware Descent Algorithm}
\end{algorithm}

Based on Lemma (\ref{lemma1}), and the associated discussion, we know that for a typical network setting, the objective function $f_0$ in problem (P) is monotonically increasing in $s_i$. In this case, we can therefore set up an iterative algorithm for solving subproblem (P1) to find the optimal $\boldsymbol{s}$, by sequentially changing $s_i$ by some $\Delta s_i$ for each user until a minimum objective is reached where all constraints of the original problem (P) are satisfied. The main constraint that is affected by decreasing $\boldsymbol{s}$ is the total latency, which is increasing with smaller $\boldsymbol{s}$. Thus we propose  a latency-aware descent algorithm, based on the standard descent-method but with modified stopping criteria. We compare two descent methods to find the optimal $\boldsymbol{s}$: the gradient-descent method and the Newton method. For both descent methods applied, we implement the structure for the outer optimization algorithm as described in Algorithm 1, in which the modified stopping criterion with latency awareness is unique to this algorithm design and is crucial in making sure the optimized solutions meet the latency constraint.

The outer algorithm works as follows. We initialize $0 < s_{i,0} < u_i$, input the simulation parameters, and update the step or search direction $\Delta \boldsymbol{s}$ as in standard gradient-descent ($\Delta {s_i} = - \nabla_{s_i} f_0(s_i)$) or Newton method ($\Delta {s_i} = - \nabla^2_{s_i} f_0(s_i)^{-1} \nabla_{s_i} f_0(s_i)$)~\cite{Boyd2004}. We then execute the inner algorithm for finding the optimal time and frequency allocation for the given value of $\boldsymbol{s}$. Next we proceed to the sequential update of $s_i$. For both descent methods, we use backtracking line search to find the step-length at the $k^{th}$ iteration as the vector $\boldsymbol{t}^{(k)}$, with $t_i^{(k)}$ as the step-length for the $i^{th}$ user, and update the offloaded bits for the next iteration as $s_i^{(k+1)} = s_i^{(k)} + t_i \Delta s_i$. We then check the stopping criteria for convergence of the outer algorithm. In this step, we introduce a novel modification to the classical stopping criterion for descent methods, which is necessary to arrive at the optimal solution for the original problem (P) as shown in the next proposition.

\begin{prop}
Given that condition (\ref{gradf0}) in Lemma 1 holds, a latency-aware termination of the descent algorithm is necessary to reach an optimal solution satisfying all constraints of the original problem (P). The latency based stopping criterion is given as
\begin{equation}\label{timestop}
T_{\text{total}} = \max \ (t_{u,i} + t_{L,i}, \sum_{j=1}^{3} T_j) \leq T_d 
\end{equation}
\end{prop}
\begin{proof}
We have two stopping criteria for the algorithm termination; the first is specific to the descent method applied and is defined by the suboptimality condition, $f_0(x) - p^\star \leq \epsilon_1$, where $p^\star$ is the optimal solution~\cite{Boyd2004}, while the second is as given in (\ref{timestop}) derived from the delay constraints (c-d) in (\ref{Pnew}).

For the considered system-level energy minimization problem, where $w \neq 0$, Lemma~\ref{lemma1} always holds and the primal objective function $f_0(s_i)$ is a monotonically increasing function in $s_i$ as shown on the left y-axis in Figure~\ref{stop_time} for the $i^{th}$ user. The energy is hence minimized as $s_i$ approaches zero. While this may be optimal for smaller data requests $u_i$ such that all data is computed locally, for larger data requests, however, computing all data locally can be time-inefficient. This is because the time for local computation which is proportional to $q_i = u_i-s_i$ as in (\ref{t_Li}) increases linearly with $q_i$ (or equivalently, it increases linearly as $s_i$ decreases) and may exceed the delay constraint, that is $\max(t_{u,i} + t_{L,i}) > T_d$ as shown in Figure \ref{stop_time} on the right y-axis. Here for small $s_i$, the total time $T_{\text{total}}$ exceeds the latency constraint, due to large $q_i$. For such scenarios, the optimal $s_i$ is then found as the point where the time constraint is met with equality, that is, $T_{total} = T_d$ as shown, where $s_i^\star$ becomes the amount of data that can be offloaded such that the system's energy consumption is minimized within the delay constraint. Hence latency-aware termination of the descent algorithm in the outer optimization problem becomes necessary for large data requests, such that the delay constraints for the original problem (P) are satisfied. 
\end{proof}
\setlength{\belowcaptionskip}{-20pt}
\begin{figure}[!tbp]
  \centering
  \begin{minipage}[b]{0.45\textwidth}
    \includegraphics[scale = 0.45]{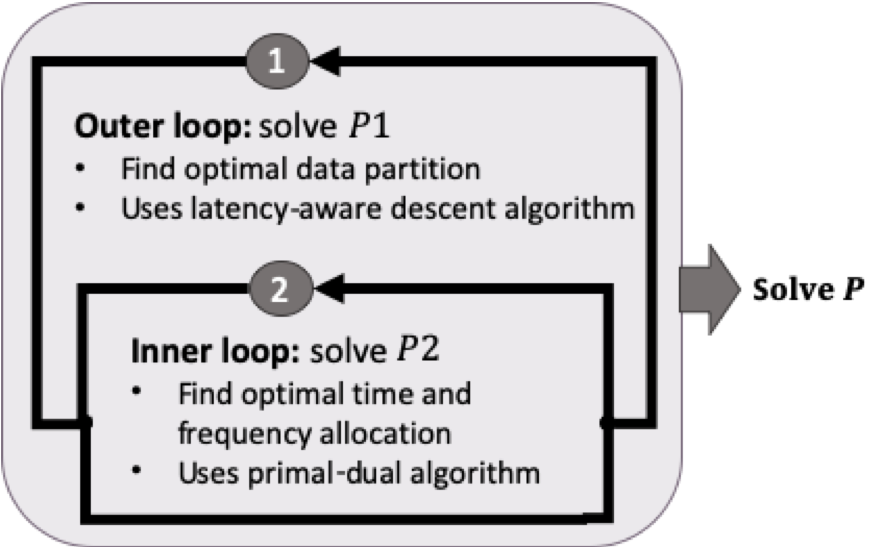}
    \caption{Nested algorithm architecture for the solution of (P)}
    \label{nested_algo}
  \end{minipage}
  \hfill
  \begin{minipage}[b]{0.5\textwidth}
    \includegraphics[scale = 0.45]{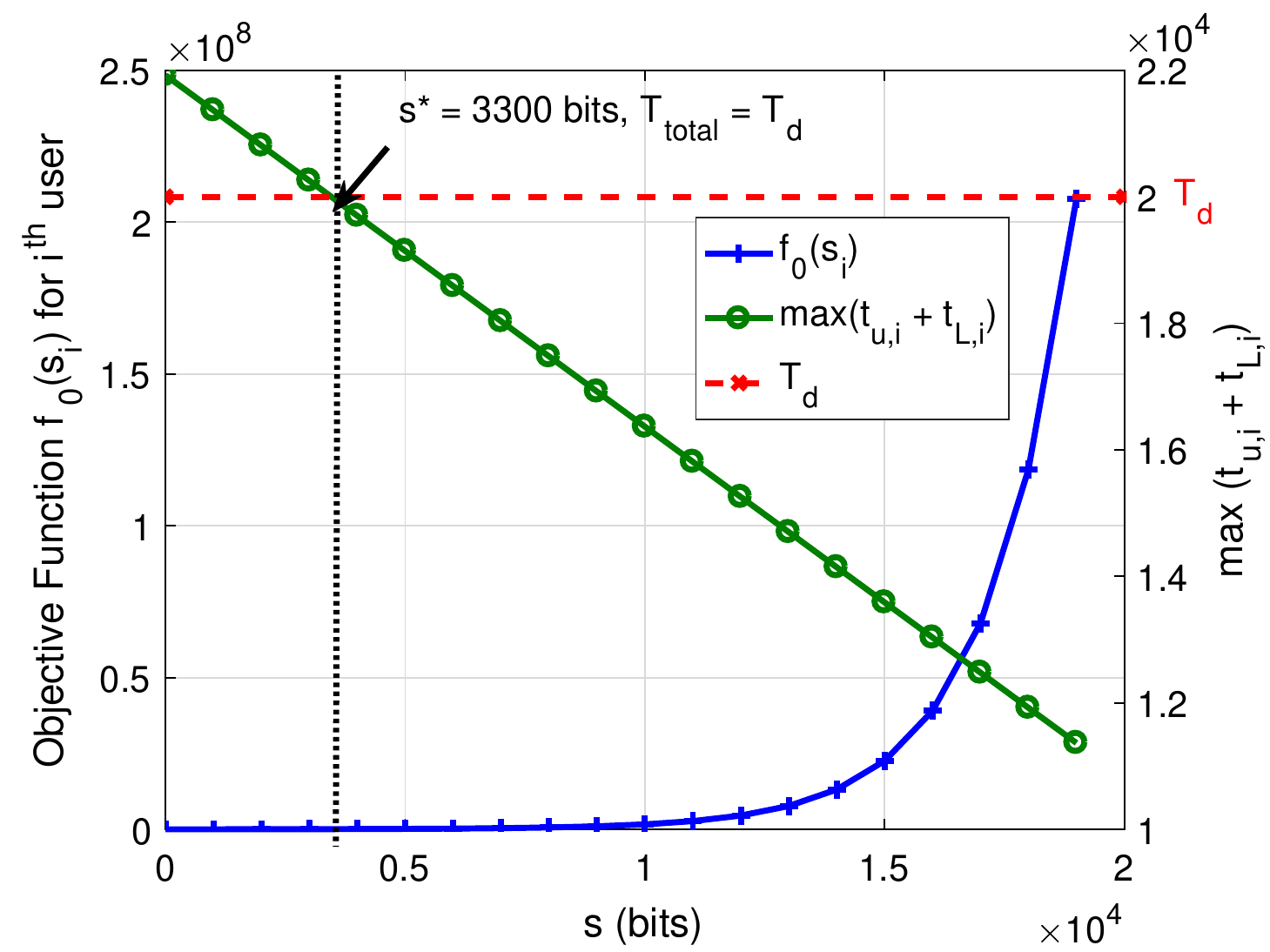}
    \caption{Primal objective function monotonically increases w.r.t $s_i$}
    \label{stop_time}
  \end{minipage}
\end{figure}

\subsection{A Primal-Dual Algorithm for Inner Optimization}
Based on Lemma (\ref{lemma2}), problem (P) is convex for a fixed $\boldsymbol{s}$. Consider the subproblem (P2) to solve for $\boldsymbol{t,f}$; this problem is convex and we can show that strong duality holds since Slater's condition is satisfied, that is, we can find a strictly feasible point in the relative interior of the domain of the problem where the inequality constraints hold with strict inequalities. Therefore, to solve for (P2), we formulate a primal-dual problem using the Lagrangian dual method. At each iteration of the outer algorithm discussed above, that is for a fixed $s_i$, we solve a primal-dual problem for the remaining variables in (P2) using Lagrangian duality analysis. 

Theorem \ref{theorem1} below provides the optimal time allocation for uplink and downlink transmissions in terms of the dual variables. It provides a solution for the time required per user to offload data to the MEC and the time consumed by the MEC to compute each user's tasks.
\setcounter{equation}{17}
\begin{theorem}\label{theorem1}
The offloading and downloading time, $t_{u,i}$ and $t_{d,i}$ respectively, can be obtained as
\begin{align*}\label{x1x2}
    &t_{u,i} = \left(\frac{\nu B}{s_i \ln 2} \left(W_0 \left(\frac{\beta_i + \xi_i}{(1 - w)} \left( \frac{N \gamma_i}{\sigma_{1,i}^2 \Gamma_{1} e}\right) - \frac{1}{e}\right) + 1 \right)\right)^{-1} \tag{\theequation a}
 \end{align*}
 \begin{align*}
    &t_{d,i} = \left(\frac{B}{\mu s_i \ln 2} \left(W_0 \left(\frac{\phi_i}{w} \left( \frac{N \gamma_i}{\sigma_{2,i}^2 \Gamma_{2} e}\right) - \frac{1}{e}\right) + 1 \right)\right)^{-1} \tag{\theequation b}
\end{align*}
Here $\xi_i$, $\beta_i$ and $\phi_i$ are the dual variables associated with the constraints (d), (e) and (g) for problem (P) in~(\ref{Pnew}) respectively.
\end{theorem}
\begin{proof}
Applying Karush-Kuhn-Tucker (KKT) conditions with respect to offloading time $t_{u,i}$ and downloading time $t_{d,i}$, respectively, we obtain equations of the form
\begin{align*}\label{KKT_time}
&(1 - w) \left( f(x_{1,i}) - x_{1,i} f'(x_{1,i}) \right) + \beta_i + \xi_i = 0 \\
&w (f(x_{2,i}) - x_{2,i} f'(x_{2,i})) + \phi_i = 0
\end{align*}
where $x_{1,i} = \frac{1}{t_{u,i}}$, $f(x_{1,i}) = \frac{\big(2^{\frac{s_i}{\nu t_{u,i} B}} - 1 \big) \Gamma_{1} \sigma_{1,i}^2}{N \gamma_i}$, $x_{2,i} = \frac{1}{t_{d,i}}$ and $f(x_{2,i}) = \frac{\big(2^{\frac{\mu s_i}{t_{d,i} B}} - 1 \big) \Gamma_{2} \sigma_{2,i}^2}{N \gamma_i}$. For the function of the form $f(x) = \sigma^2 (2^{\frac{x}{B}} - 1)$ and $y = f(x) - x f'(x)$ of $x > 0$, its inverse can be shown to be obtained from the principal branch of the Lambert $W$ function, $W_0$ as~\cite{lambert1996}
\begin{equation}\label{LambertSol}
x = \frac{c B}{\ln 2} \left(W_0 \left(\frac{-y}{\sigma^2 e} - \frac{1}{e}\right) + 1 \right)
\end{equation}
Details of the proof provided in Appendix B.
\end{proof}
We now proceed to deriving the optimum frequency allocation. Theorem \ref{theorem2} below provides a solution for the each user's CPU frequency for local data computation, and the frequency allocated at the MEC for coomputing each user's offloaded tasks. It is worth mentioning that the optimal primal solution derived in Theorems \ref{theorem1} and \ref{theorem2} is specific for the considered system. Thus the proposed primal-dual approach for solving (P2) reveals the optimal solution structure which would otherwise be obscured by plugging into a standard convex solver.

\begin{theorem}\label{theorem2}
The optimal CPU frequency at the user ($f_{u,i}$) and at the MEC ($f_{m,i}$) can be obtained in closed form from the cubic equations below

\vspace{2mm}
\begin{subequations}
\begin{minipage}{0.4\textwidth}
  \begin{equation}
    \label{difffi}
      f_{u,i}^\star = \left( \frac{\xi_i}{2 (1 - w) \kappa_i} \right)^{\frac{1}{3}}
  \end{equation}
\end{minipage}
\begin{minipage}{0.6\textwidth}
  \begin{equation}
   \label{difffmi}
    2 w \kappa_m d_m s_i f_{m,i}^3  + \lambda_5 f_{m,i}^2 - \theta_i d_m s_i = 0
  \end{equation}
  \end{minipage}
  \vspace{0.00001mm}
\end{subequations}

\noindent where $\theta_i$ and $\lambda_5$ are the dual variables for constraints (f) and (h) respectively.
\end{theorem}
\begin{proof}
Obtained directly by applying KKT conditions with respect to $f_{u,i}$ and $f_{m,i}$. The chosen root for the cubic equation is that which satisfies the boundary conditions. See Appendix C.
\end{proof}

Theorems \ref{theorem1} and \ref{theorem2} provide the optimal solution of the primal variables in terms of the dual variables. We can use them to design a primal-dual algorithm to solve the convex optimization problem (P2) in (\ref{Pnew}) for a fixed $\boldsymbol{s}$. The dual-function for this problem can be defined as
\begin{equation}\label{DualFn}
    g(\lambda_1,\lambda_5, \boldsymbol{\beta, \phi, \xi, \theta}) = \underset{\boldsymbol{t,f}}{\inf} \mathcal{L}(\boldsymbol{t,f}, \lambda_1,\lambda_5, \boldsymbol{\beta, \phi, \xi, \theta})
\end{equation}
where $\mathcal{L}$ is the Lagrangian for problem (P2) defined in (\ref{LDual}) and the dual-problem is given as
\vspace{-2mm}
\setcounter{equation}{21}
\begin{align*}\label{PDual}
    \text{P-dual: }\max \ &g(\lambda_1,\lambda_5, \boldsymbol{\beta, \xi, \theta, \phi})\\
    \text{s.t. } &\lambda_1, \lambda_5, \beta_i, \phi_i, \xi_i, \theta_i \geq 0 \ \forall i = 1...K \tag{\theequation}
\end{align*}

The Lagrangian dual $\mathcal{L}$ has no closed-form, so we use a subgradient approach to solve the dual minimzation problem~\cite{Boyd2003}. We design a primal-dual algorithm which iteratively updates the primal and dual variables until reaching a target accuracy. We use the optimal primal solutions in Theorems \ref{theorem1} and \ref{theorem2} to obtain the dual function, $g(\boldsymbol{x})$, as given in (\ref{DualFn}). The problem then becomes maximizing this dual function in terms of the dual variables. The subgradient terms with respect to all dual variables are as follows.
\setcounter{equation}{22}
\begin{align*}\label{subgrads}
&\nabla_{\lambda_1}\mathcal{L} = \sum_{j=1}^{3} T_j - T_{\text{delay}}, \ \ \nabla_{\beta_i}\mathcal{L} = t_{u,i} - T_1, \ \  \nabla_{\xi_i}\mathcal{L} = \frac{c_i q_i}{f_{u,i}}  + t_{u,i} - T_{\text{delay}} \tag{\theequation a-c}
\end{align*}
\begin{align*}
&\nabla_{\theta_i}\mathcal{L} = \frac{d_m s_i}{f_{m,i}} - T_2, \ \ \ \ \ \ \nabla_{\phi_i}\mathcal{L} = t_{d,i} - T_3, \ \ \nabla_{\lambda_5}\mathcal{L} = \sum_{i=1}^{K} f_{m,i} - f_{m,\max}\tag{\theequation d-f}
\end{align*}

For our implementation, we update the dual variables based on the shallow-cut ellipsoid method using the sub-gradient expressions in (\ref{subgrads}-f). The sub-gradient in the ellipsoid algorithm is calculated at the ellipsoid center, $x = (\lambda_1,\lambda_5, \boldsymbol{\beta, \xi, \theta, \phi})$, to reach the minimum volume ellipsoid containing the minimizing point for the dual-function $g(\lambda_1,\lambda_5, \boldsymbol{\beta, \phi, \xi, \theta})$. For each iteration, the primal variables updates are based on Theorems \ref{theorem1} and \ref{theorem2}, and a new value for $g(\lambda_1,\lambda_5, \boldsymbol{\beta, \phi, \xi, \theta})$ is calculated. Since the ellipsoid algorithm is not a descent method, we keep track of the best point for the dual function $g(\lambda_1,\lambda_5, \boldsymbol{\beta, \phi, \xi, \theta})$ in (\ref{DualFn}) at each iteration of the inner algorithm. These primal-dual update steps are repeated until the desired level of precision, $\epsilon_2$, is reached for the stopping criterion, which is the minimum volume of the ellipsoid in our algorithm. The steps for the primal-dual algorithm are shown in Algorithm 2.

\begin{algorithm}[t]
\caption{Solution for Inner Optimization Problem (P2)}
\textbf{Given} a starting point $\boldsymbol{s}$\\
\textbf{Initialize:} Dual variables, $\lambda_1,\lambda_5$, $\beta_i, \xi_i, \theta_i, \phi_i \forall i$.\\
\textbf{Begin Primal-Dual Algorithm for (P2)}
\begin{itemize}[leftmargin=*]
\item Calculate $f_{u,i}^\star $ and $f_{mi}^\star \ \forall i = 1...K$ from (\ref{difffi}) and (\ref{difffmi}) respectively. For any $i^{\text{th}}$ user, \\
$\circ$ \textbf{If} $f_{mi}^\star < f_{m,\min}$, apply boundary condition, \textbf{then} $f_{mi}^\star = f_{m,\min}$ \\
$\circ$  \textbf{If} $f_{u,i}^\star < f_{\min}$, OR $f_{u,i}^\star > f_{\max}$, \textbf{then} apply boundary conditions, $f_{u,i}^\star = f_{\min}$ OR $f_{u,i}^\star = f_{\max}$. 
\item Calculate $t_{u,i}$ and $t_{d,i}$, using (\ref{x1x2}-b). Then $T_1^\star = \max t_{u,i}^\star$ and $T_3^\star = \max t_{d,i}^\star$.
\item Update $p_i$ and $\eta_i$ using (\ref{poweralloc}). 
\item Using updated power values to calculate $\sigma_{1,i}^2$ and $\sigma_{2,i}^2$.
\item Calculate $t_{MEC}^\star$ from (\ref{tMEC}). Then $T_2^\star = \max t_{MEC}^\star$
\item Find dual function value, $g(\lambda_1,\lambda_5, \boldsymbol{\xi, \theta, \beta, \phi})$, in (\ref{DualFn})\\
$\circ$ If dual variables converge with $\epsilon_2$, \textbf{stop}\\
$\circ$ Else, find subgradients in (\ref{subgrads}-f), update dual-variables using ellipsoid method, \textbf{continue}
\end{itemize}
\textbf{End Primal-Dual Algorithm for (P2)}
\end{algorithm}

\subsection{Algorithm Implementation and Convergence}\label{algo_converg}
The nested algorithm for optimally solving (P) is comprised of the outer latency-aware descent algorithm, and the inner subgradient based primal-dual algorithm. These algorithms are executed at the MEC server which then distributes results to the users. Initializing the nested-algorithm only requires the user locations and the amounts of data requested, which can be shared over a control channel before the actual data communication over data traffic channels in uplink and downlink takes place. For simulations in this paper, we implemented the algorithm on a personal computer, but implementation on MEC servers with high computational capabilities can be expected to run seamlessly in a wireless fading environment.

For the inner algorithm, we use the shallow-cut ellipsoid method to update the dual variables, where convergence is guaranteed due to the convexity of problem (P2) as shown in Lemma~\ref{lemma2}. As the optimal values for the dual variables are reached in the inner algorithm, the values for the primal variables also converge to their
respective optimal values by strong duality. For the ellipsoid method, the number of iterations is proportional to the number of constraints $n$~\cite{Bland1981} since the ellipsoid volume decreases as a geometric series whose ratio depends on the dimension of the space~\cite{Goffin1983}. The convergence speed for the ellipsoid algorithm is proportional to $R \exp(-\frac{K_{\text{in}}}{2n^2})$, where $K_{\text{in}}$ is the number of iterations to reach $\epsilon_2$-optimal solution for (P2) \cite{Bubeck2015}, which requires modest computation per step of $\mathcal{O}(n^2)$~\cite{Boyd2008}.
 
The convergence of the outer algorithm depends on the descent method chosen, i.e. the gradient descent or Newton method, and the line-search method. We use the inexact backtracking line-search in our latency-aware descent algorithm due to its simplicity and effectiveness which is known to always terminate~\cite{Boyd2004}. For the standard gradient-descent method, $f_0(s_i(k))$ converges to the optimal point $p^\star$ linearly, while the Newton method has a linear start and then hits the quadratic convergence after a small number of iterations~\cite{Boyd2004}. While the Newton method can warrant faster convergence with a significantly lower number of iterations, it has higher computation cost with each Newton iteration requiring $\mathcal{O}(n^3)$ flops compared to $\mathcal{O}(n)$ flops required for gradient descent~\cite{Ryan2019}. In our latency-aware descent outer algorithm, since we add an additional stopping criterion based on the latency, the algorithm may stop earlier than the standard implementation. Therefore, we expect to see the latency-aware gradient descent to have the same linear convergence, but the latency-aware Newton method may not hit the quadratic convergence if the latency constraint is met before that.

Among all the methods employed, the convergence speed of the ellipsoid algorithm is the slowest component of the nested algorithm. The convergence speed varies with the number of constraints $n$, which is especially relevant for systems with large number of users since constraints (d)-(g) for problem (P) in (\ref{Pnew}) are per-user constraints. For a fixed number of users and consequently for a fixed number of dual-variables the rate of convergence for the ellipsoid algorithm is linear (similar to the center-of-gravity method~\cite{Bubeck2015} upon which the ellipsoid algorithm is based~\cite{Bland1981}) albeit typically at a much slower rate than gradient descent. 

At each iteration of the descent algorithm, a call is made to the ellipsoid algorithm (see Algorithm 1), and the outer algorithm moves on to the next iteration after convergence of the inner algorithm is reached for a given $\boldsymbol{s}$. Therefore, the overall number of iterations $K_{\text{tot}}$ is a product of the number of iterations $K_{\text{out}}$ and $K_{\text{in}}$ of the outer and inner algorithms respectively. For the nested algorithm with gradient descent, the convergence speed is linear, where the convergence time is dominated by the inner ellipsoid algorithm. For the nested algorithm with Newton descent method, the convergence is super-linear and sub-quadratic, with the convergence speed being closer to linear than quadratic due to the slow speed of the ellipsoid algorithm compared to the Newton method. The convergence of the nested algorithm would therefore be significantly faster when using the latency-aware Newton method compared to gradient-descent.  
\section{Numerical Results}
\setlength{\belowcaptionskip}{-15pt}
\begin{figure}[!tbp]
  \centering
  \begin{minipage}[b]{0.4\textwidth}
    \includegraphics[scale=0.4]{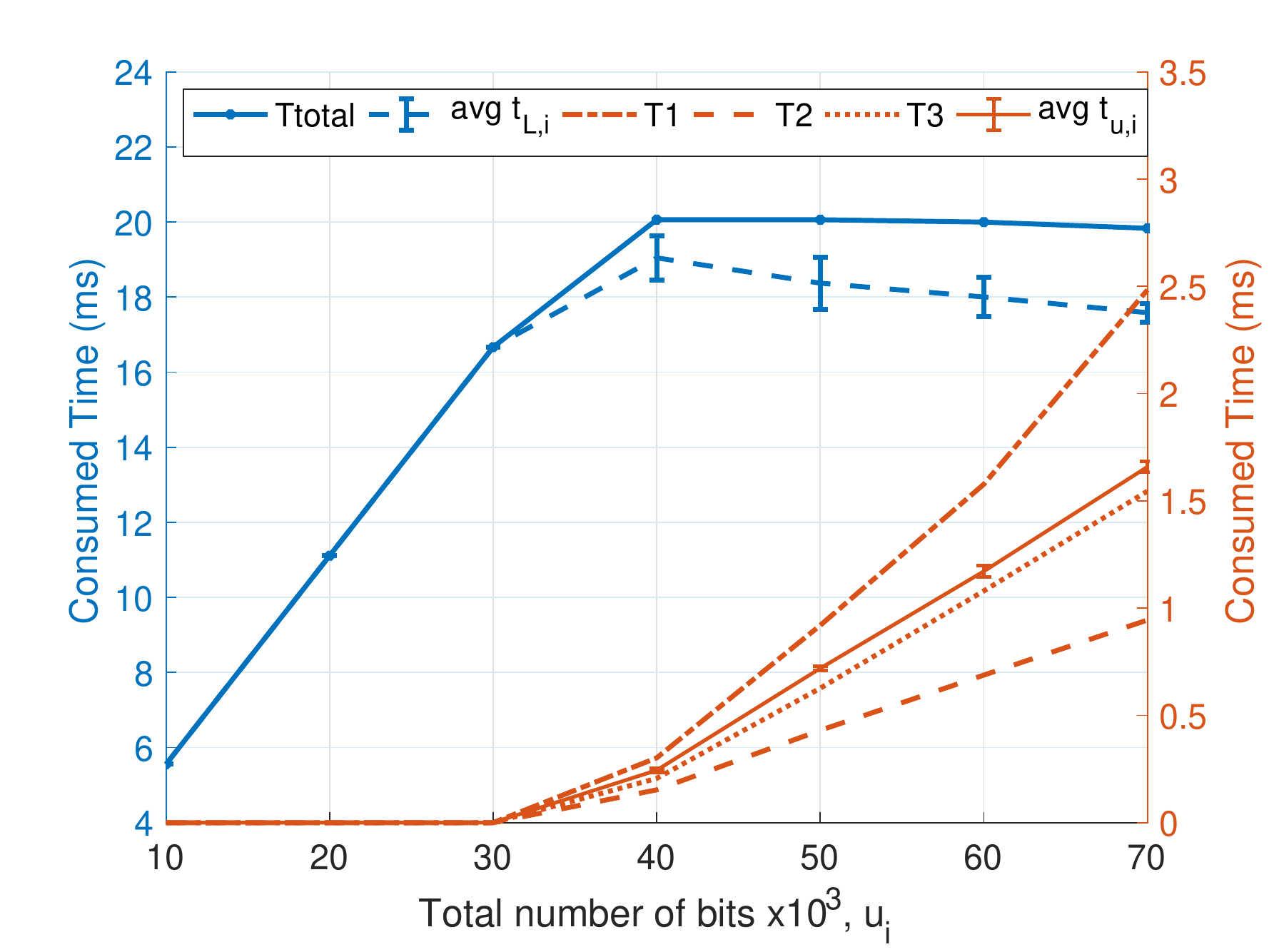}
    \caption{Consumed total time and the percentage of time spent in each phase.}
    \label{time_figure}
  \end{minipage}
  \hfill
  \begin{minipage}[b]{0.5\textwidth}
    \includegraphics[scale=0.4]{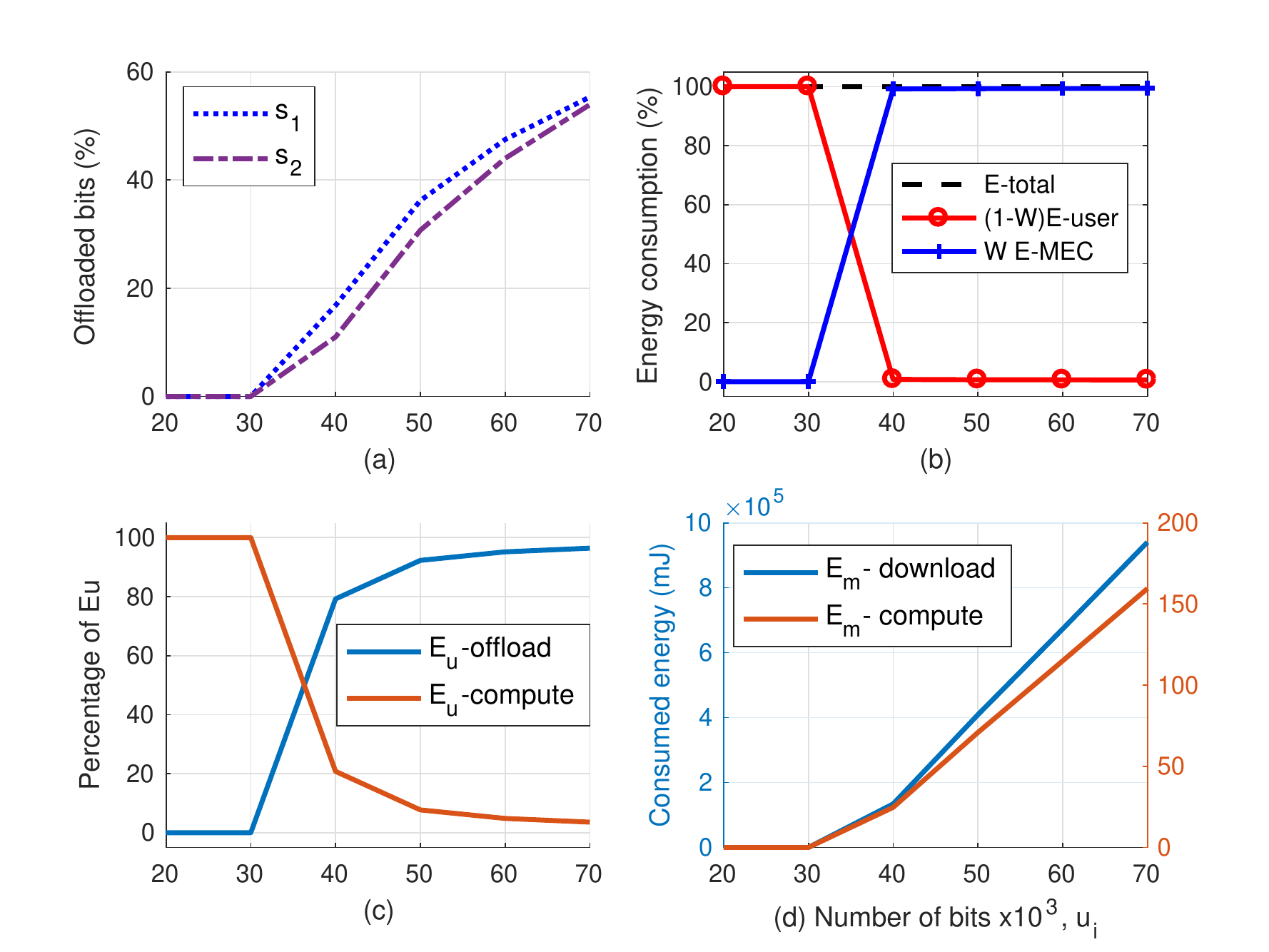}
    \caption{Percentage offloaded data (a), total energy consumption (b), transmission/computation energy consumption at users (c) and MEC (d)}
    \label{bits_all}
  \end{minipage}
\end{figure}

In this section, we evaluate the solution of energy minimization problem (P) with respect to energy consumption, time required, and the partition of bits offloaded to the MEC for computation. We consider a $20\text{m}\times20\text{m}$ area with 4 APs and 16 users randomly located with $K = 4$ users per AP's coverage area and $N = 100$ as shown in Figure~\ref{system_model}. For simulations, $w = 10^{-3}$, $T_d = 20$ms (in accordance with the AR/VR applications requirement for motion-to-photon latency~\cite{Intel2017}), $B = 5$MHz, $\tau_c = B T_d$, $\Gamma_1 = \Gamma_2 = 1.25$, $\mu = 2$, $\kappa_i = 0.5$pF, $\kappa_m = 5$pF, $c_i = 1000$, $d_m = 500$, $\gamma = 2.2$, $\sigma = 2.7$dB, $\sigma_r^2 = -127$dBm, $\sigma_k^2 = -122$dBm, $(f_{\min}, f_{\max}, f_{m,\min}) = (60,1800,2200)$ MHz. Each MEC processor has 24 cores with maximum frequency of $3.4$GHz. For initialization of Algorithm 1, any feasible value $0 \leq s_i \leq u_i$ can be chosen which satisfies the constraints for power and latency. For our numerical simulations we start with $s_i  = 0.6 u_i \ \forall i$. Transmit power available at user and AP is 23 dBm and 46 dBm respectively. To calculate the interference and noise power in (\ref{INul}) and (\ref{INdl}) which include massive MIMO pilot contamination and intercell interference, we assume that user terminals transmit at their maximum power, that is $p_{qi} = 23$dBm, and the interfering APs use equal power allocation in the downlink, that is $\eta_{qi} = \frac{1}{K} \ \forall i$. Numerical results are averaged over 1000 independent channel realizations of $\mathbf{H}$ and $\mathbf{G}$.
\vspace{-1mm}
\subsection{Effect of the Amount of Data for Computation}
Figure~\ref{time_figure} shows the total time consumption and time consumed per phase as the amount of requested data increases. We use $u_i = u \ \forall i \in [1,K]$. For low data requests, $u < 40$kbits, the total time consumption is always less than $T_d$. For $u > 40$kbits, however, the consumed time becomes a limiting factor and the energy is minimized such that the latency constraint is met with equality. Here $T_{total}$ is as given in (\ref{timestop}). We also show the breakdown for time consumed in each phase, where Phase II consumes the minimum time due to the MEC's high CPU frequency. The offloading time $T_1$ is more than the downloading time $T_3$ due to the difference in user and AP transmit powers even though we assume that the computed results $\tilde{s}_i = 2 s_i$. The average time for local computation at users is much higher than $T_2$ because of lower processing speed at the users. For larger data requests, the time $\max(t_{u,i}+t_{L,i})$ in (\ref{timestop}) becomes equal to $T_d$ and leads to the termination of the latency-aware descent algorithm.
\setlength{\belowcaptionskip}{-20pt}
\begin{figure}[t]
        \begin{minipage}[b]{0.4\textwidth}
                \centering
                \includegraphics[scale = 0.35]{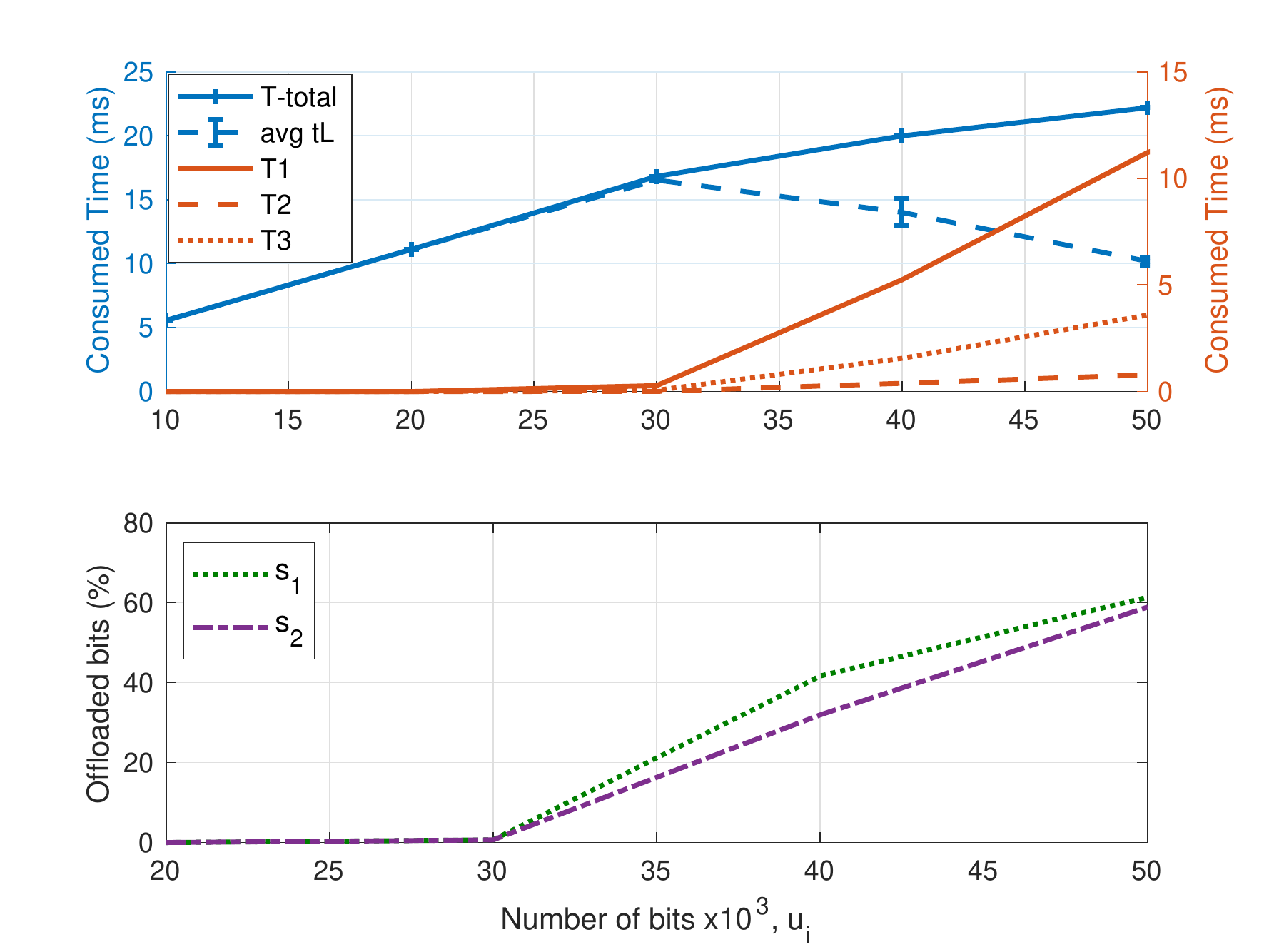}
                \caption{Consumed time, total and per phase (top), percentage offloaded data (bottom) under imperfect CSI}
                \label{time_error}
        \end{minipage}
       \hfill
        \begin{minipage}[b]{0.5\textwidth}
                \centering
                \includegraphics[scale = 0.38]{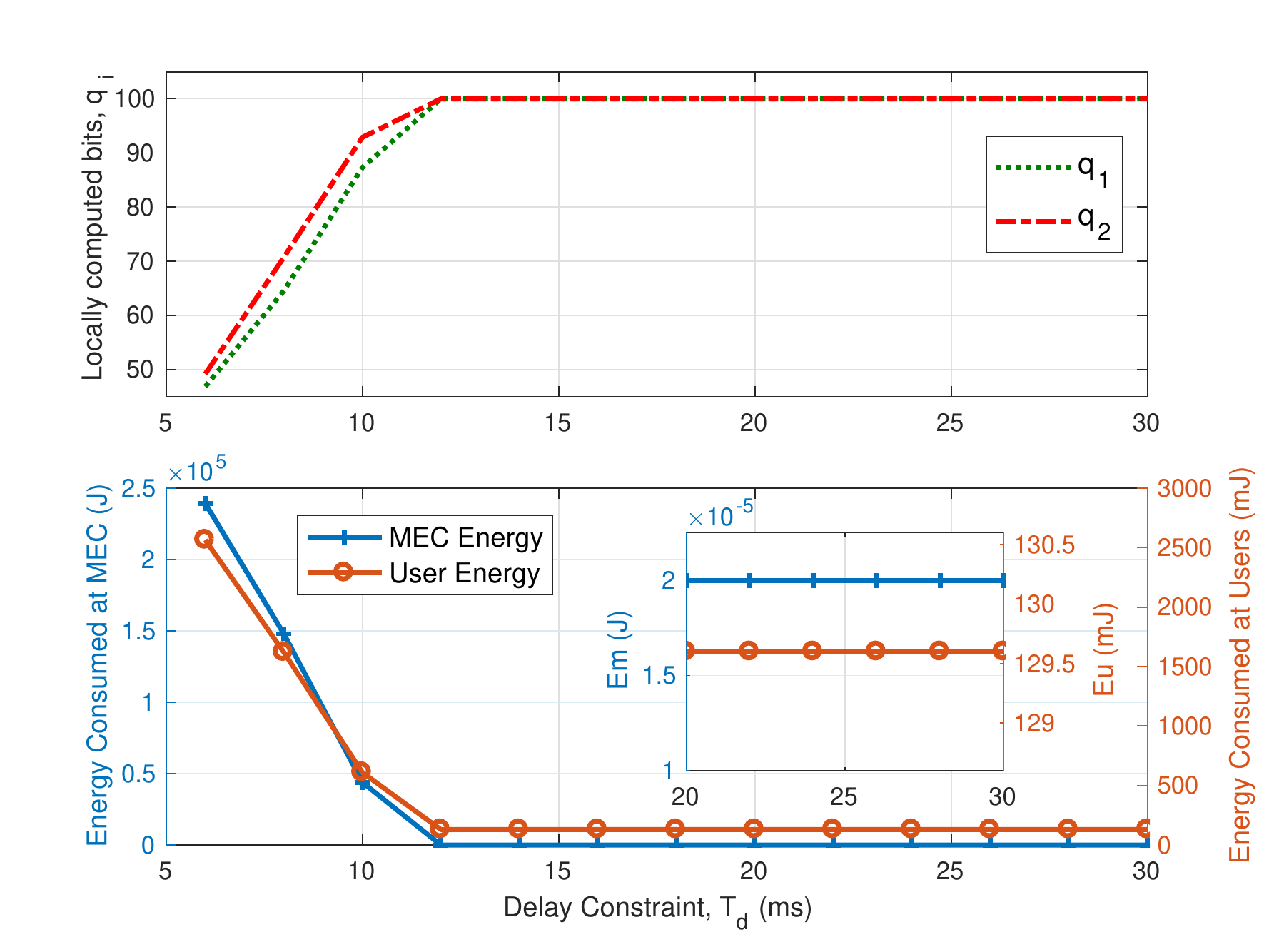}
		\caption{Percentage of locally computed data (top) and energy consumption at users and MEC (bottom)}
		\label{delay_new}
        \end{minipage}
\end{figure}

Figure~\ref{bits_all} shows the percentage of offloaded bits (for two representative users), the percentage energy consumption at the users and at the MEC, the breakdown for the percentage of energy consumed at the users, and the actual energy consumed at the MEC for computation and transmission. When $u < 30$kbits, no data is offloaded, hence the system's energy consumption is solely due to the user computation. As data requests increase, some data is offloaded to the MEC reaching around 60\% at $u_i = 70$kbits. Correlating with Figure~\ref{time_figure}, the data partition for local and offloaded bits is optimized such that the total time remains within the latency constraint. For $u > 40$kbits, the MEC's energy consumption becomes dominant since more than half the data is offloaded. Note that the actual energy consumption at both users and the MEC would increase as more data is requested since $E_{u}$ and $E_{m}$ are both proportional to $u_i$. Computation energy at user is proportional to $u_i - s_i$ and hence decreases with increasing $s_i$. For the MEC, since both computation and transmission energy in (\ref{E_OC}) and (\ref{E_dl}) respectively, increase proportionally to $s_i$, their percentage energy consumption remains constant. We therefore show the breakdown of the actual energy consumption at the MEC. For both the users and the MEC, wireless transmission consumes significant energy. For the MEC, since downloading energy (left yaxis) is significantly larger than the computation energy (right yaxis), they are therefore shown on separate axes.
\vspace{-1mm}
\subsection{Effect of Latency Requirement}
Figure~\ref{delay_new} shows the percentage of locally computed data (for two users), and the breakdown for the total energy consumed at the users and MEC as $T_d$ is increased for $u_i = u = 20\text{kbits} \ \forall i$. For strict latency constraint, we see that less than half the data is computed locally. This implies that energy is consumed for wireless transmission (offload/download) as well as computation (at both users and MEC). However, as the delay requirement is relaxed, for $T_d > 12$ms, all data is computed locally, and the weighted energy consumption in this regime can be approximated as  $E_{\text{total}} \approx (1 - w)E_u$. With the relaxed constraint, users can afford to spend more time for computation as $t_{L,i}$ using their low CPU frequencies, which leads to lower energy consumption. We therefore see that both $E_u$ and $E_m$ settle to a constant level. $E_m$ becomes negligible since all data is computed locally for larger values of $T_d$.
\vspace{-1mm}
\subsection{Effect of Data Partition}  
\setlength{\belowcaptionskip}{-20pt}
\begin{figure}[t]
        \begin{minipage}[b]{0.45\textwidth}
                \centering
                \includegraphics[scale = 0.35]{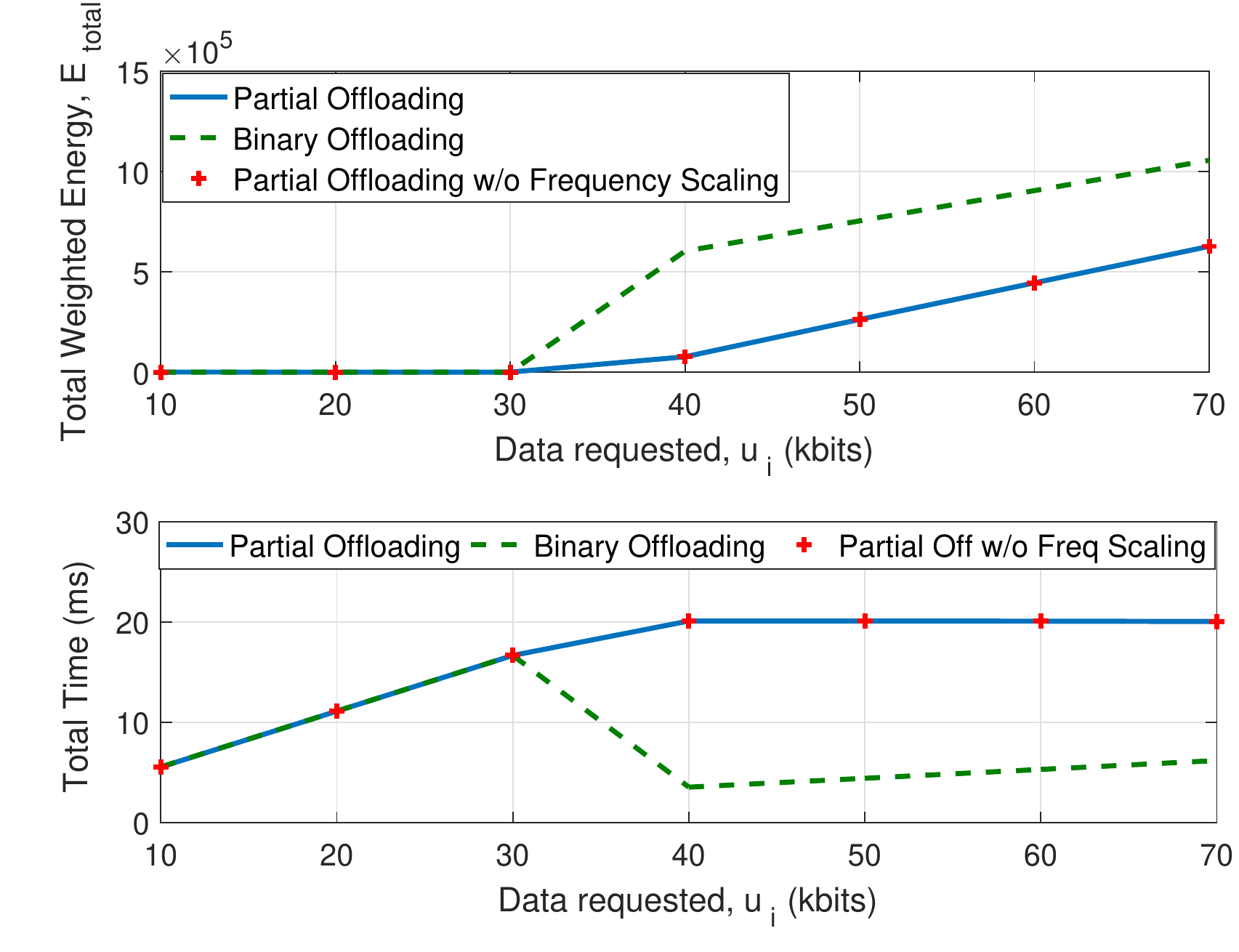}
                \caption{Energy and time consumption for three different schemes with increasing amount of requested data}
                \label{both_data}
        \end{minipage}
       \hfill
        \begin{minipage}[b]{0.45\textwidth}
                \centering
                \includegraphics[scale=0.35]{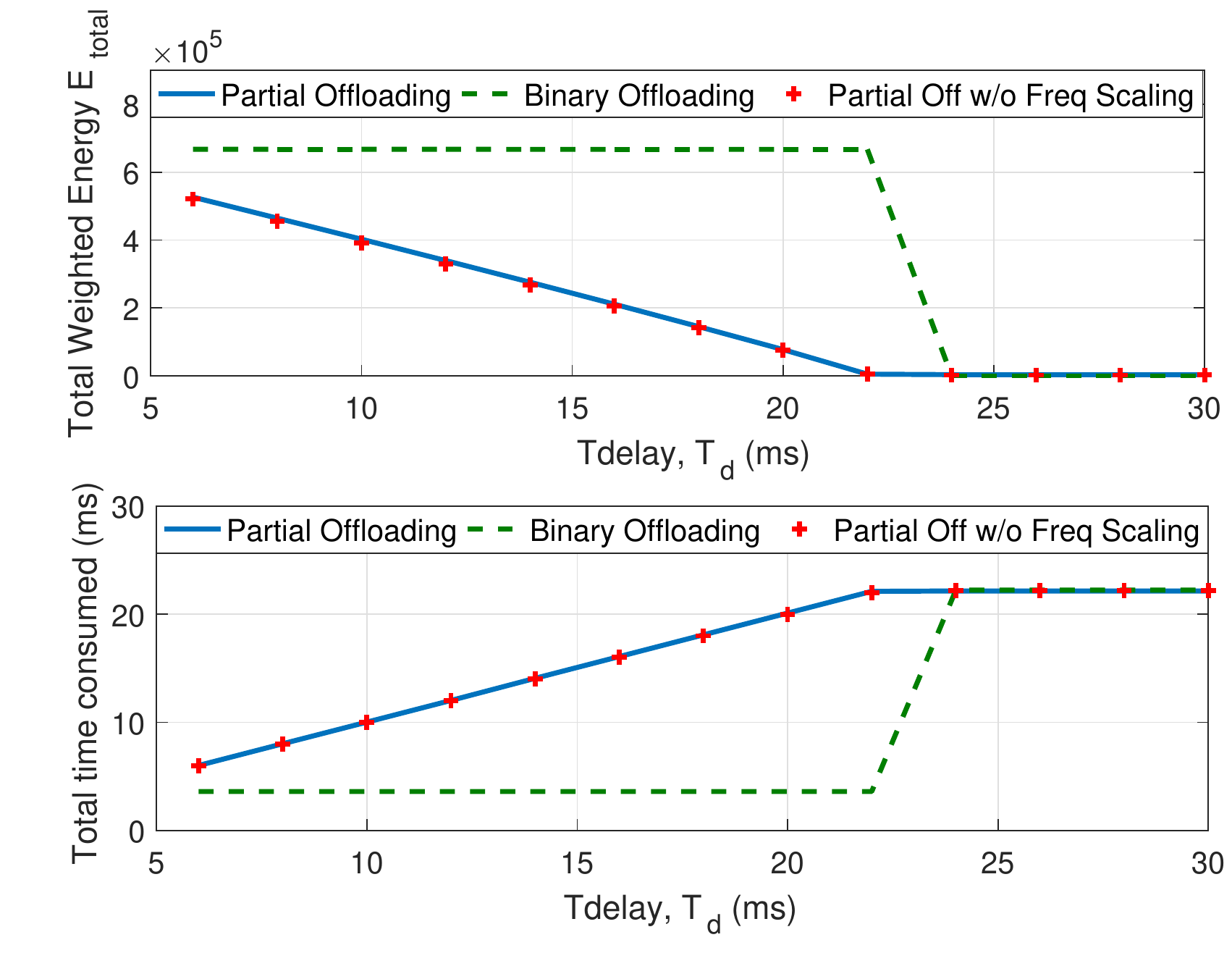}
                \caption{Energy and time consumption for three different schemes with relaxation in latency constraint}
                \label{both_Tdelay}
        \end{minipage}
\end{figure}
In this section we analyze the effect of data partitioning on the energy consumption under a given latency constraint. Specifically we compare the proposed partial offloading scheme with binary offloading where each user's task cannot be partitioned and is either computed entirely at the local user or at the MEC. The binary offloading solution presented is the best one with the lowest overall energy consumption chosen from all possible binary offloading combinations. Figure \ref{both_data} shows the time and energy consumption for the case when the requested computation data is increased with a fixed latency constraint of $T_d = 20$ms. We see significant disparity between the binary and partial offloading schemes. For low data requests, local processing at users is optimal so both schemes consume the same time and energy, however for larger data requests, the binary scheme offloads all the data to the MEC, leading to faster time performance but at the expense of multiple times larger energy consumption, attributed mainly to the energy consumed for wireless transmission in phases I and III. Note that the partial offloading solutions always meet the latency constraint so there is no benefit in faster time performance for the binary scheme, whereas the energy saving for partial offloading is significant.

Figure \ref{both_Tdelay} shows the energy and time consumption of the two schemes as the delay constraint is relaxed, that is, $T_d$ is increased at a fixed amount of requested data, $u_i = 40$ kbits. For this amount of data, both the energy and time consumed for the two schemes converge for $T_d \geq 22$ms, when the latency constraint is lax enough to allow all data to be computed locally for both schemes. For tighter latency constraints, however, binary offloading consumes much higher energy for tight latency since all data is offloaded to the MEC to meet the latency requirement. Partial offloading with data partitioning therefore appears as a potent design variable for the resource allocation problem, with significant impact on the system's energy consumption.
\vspace{-1mm}
\subsection{Effect of Frequency Scaling}  
In addition, we analyze the effect of CPU frequency scaling on energy consumption and latency. Figures \ref{both_data} and \ref{both_Tdelay} show the results for partial offloading with and without frequency scaling against the amount of requested data $u_i$ and the latency requirement $T_{d}$. The scheme without frequency scaling simply allocates the CPU frequencies of all users at the maximum as $f_{u,i} = f_{\max}$ and the MEC frequency equally among all users' tasks as $f_{m,i} = \frac{f_{m,\max}}{K}$. 

The results in Figure \ref{both_data} show that frequency scaling has negligible effects on the energy and time consumption. Figure \ref{both_Tdelay} also shows comparable results of partial offloading with and without frequency scaling as the latency requirement is relaxed. These results suggest that the overhead of optimizing CPU frequency can be avoided, as long as the data partition and time allocation per phase is optimized.
\vspace{-1mm}
\subsection{Effect of Massive-MIMO Channel Estimation Error}
Figure~\ref{time_error} shows the percentage energy consumption at the users and MEC, and the time consumption in total and per phase for the case where the effect of pilot contamination in channel estimation is included. We see a similar trend as the case of perfect channel estimation in Figure~\ref{time_figure}, however, more data is offloaded to the MEC. This would imply that wireless transmission with the meager transmit power at the users consumes more time and energy. We therefore see that for $u_i \geq 30$ kbits almost half the total time is spent in the Phase I. Similarly, the time consumed for Phase III is approximately twice that for the case of perfect channel estimation. The time for MEC computation, $T_2$ is comparable for both cases since more offloaded data does not significantly increase the processing time at the MEC due to its powerful CPUs. A key difference to note under imperfect CSI is that for $u_i > 40$ kbits, the latency constraint is violated, since the sum time for offloading and local computation becomes larger than the delay constraint, $T_d \leq 20$ms.
\vspace{-1mm}
\subsection{Algorithm Convergence}
\setlength{\belowcaptionskip}{-20pt}
\begin{figure}[t]
        \begin{minipage}[b]{0.4\textwidth}
                \centering
                \includegraphics[scale = 0.375]{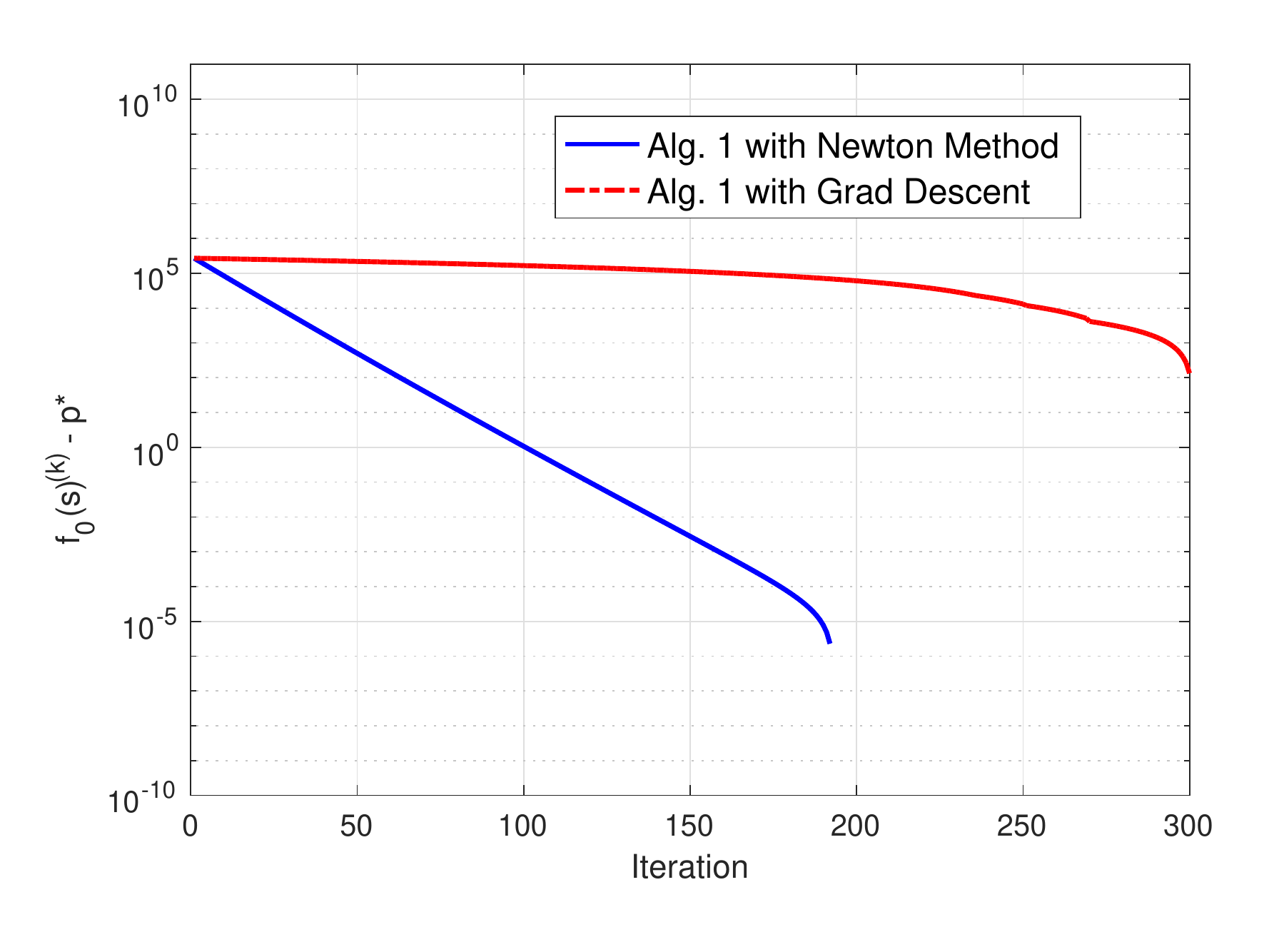}
                \caption{Algorithm Convergence}
                \label{convergence_both}
        \end{minipage}
       \hfill
        \begin{minipage}[b]{0.6\textwidth}
                \centering
                \includegraphics[scale=0.35]{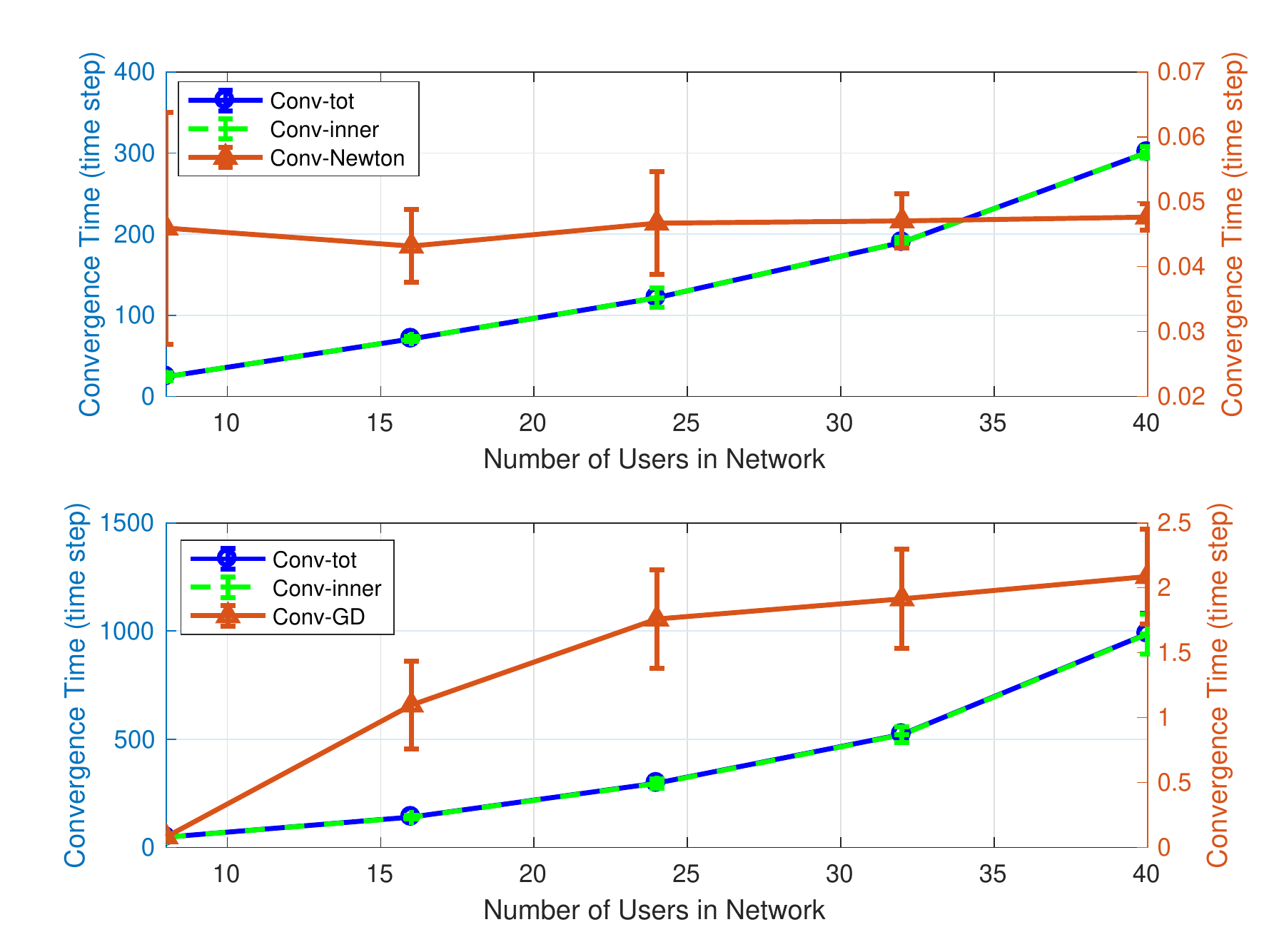}
                \caption{Convergence speed using the Newton Method (top) or GD (bottom)}
                \label{speed}
        \end{minipage}
\end{figure}
Figure \ref{convergence_both} shows the convergence for our proposed optimization algorithm with two descent methods for the outer optimization for $u_i = u = 20\text{kbits} \ \forall i$. We compare between the latency-aware Newton method and latency-aware Gradient Descent (GD). For both methods, a call is made at each iteration to the ellipsoid algorithm which has slow convergence and hence dictates the nested algorithm's speed. Using latency-aware Newton method, we observe superlinear convergence which agrees with our analysis in Section~\ref{algo_converg}. For the GD method, overall linear convergence is observed. While an outer tolerance level of $\epsilon_1 = 10^{-5}$ is used for both methods, the GD method is preemptively stopped in Figure \ref{convergence_both} because the boundary condition for $\boldsymbol{s}$ is met, that is $\boldsymbol{s = 0}$. Therefore using the latency-aware Newton method, our algorithm approaches a lower tolerance at convergence than GD.

Figure \ref{speed} shows the convergence time for the nested algorithm, and that consumed separately by the inner and outer algorithms. The overall convergence time increases almost linearly with the number of users. A significant disparity in the convergence speed when using latency-aware Newton method or GD is seen, since the Newton method converges in considerably fewer iterations compared to GD, hence making fewer function calls to the inner ellipsoid algorithm with slow convergence. Looking at the time breakdown for the outer and inner algorithms, the inner ellipsoid algorithm consumes almost all the time of the nested algorithm. Comparing between the outer algorithms, latency-aware Newton method converges more than ten times faster than GD and has an almost constant convergence speed with respect to the number of users. For latency-aware GD, however, the convergence time increases with the number of users. For our implementation on a personal computer, the time step unit is a second, however for faster machines running this algorithm, for instance MEC servers with the high-performance CPUs, this time-step could be significantly smaller and can potentially allow real time implementation.

\section{Conclusion} 
We formulated a novel system-level energy optimization problem for a delay constrained, massive MIMO enabled multi-access edge network. We designed efficient nested algorithms to minimize the total weighted energy consumption at both the user(s) and the MEC, with more weight on the users' energy consumption to commensurate the different magnitudes of available power at a user and an MEC. Comparing between two different approaches for the outer algorithm showed that the latency-aware Newton method is fast and scalable with the number of users. These algorithms demonstrate that it is optimal to compute data locally for a low amount of data requests or relaxed latency constraint. For larger data requests, however,  it is necessary to partially offload data to the MEC for computation in order to meet the latency constraint since the local computation time at users is a limiting factor due to meager processing resources. Furthermore, channel estimation error on massive MIMO links due to pilot contamination causes the transmission time and energy to increase owing to larger amounts of data offloaded to the MEC. Comparison to the binary offloading scheme also reveals significant gains in energy efficiency for the proposed partial offloading scheme. Our algorithms offer practical means to achieving the minimum network energy consumption while meeting the required latency.

\section{Appendix}
\subsection{Appendix A - Proof for Lemma 2}
The objective function is affine and convex. 
\begin{itemize}
\item Constraints (c), (e), (g), (h) for (P) in (\ref{Pnew}) are linear. 
\item For constraints (a) and (b), the second term is quadratic and convex in $f_{u,i}$ and $f_{m,i}$ respectively. The first terms are of the form $f(x) = x 2^{\frac{1}{x}}$ in $t_{u,i}$ and $t_{d,i}$ respectively, with $\nabla^2_x f(x) = \frac{2^{\frac{1}{x}}}{x^3}$ > 0 for $x > 0$, and hence $f(x)$ is convex in $x$. 
\item For constraints (d) and (f), the function is of the form $f(x) = \frac{1}{x}$ in $f_{u,i}$ and $f_{m,i}$ respectively with $\nabla^2_x f(x) = \frac{2}{x^3} > 0$ and hence convex. 
\item Relevant constraints are also linear and convex in $E_{u}$, $E_{m}$ and $T_j \ \forall j$. 
\end{itemize}
\vspace{-4mm}
\subsection{Appendix B - Proof for Theorem 1}
The Lagrangian dual of the problem (P) is given as
\setcounter{equation}{23}
\begin{align*}\label{LDual}
\mathcal{L} &=  E_{\text{total}} + \lambda_0 (E_{\text{total}} - (1 - w) E_u - w E_m ) + \lambda_1 (\sum_{j=1}^{3} T_j - T_{d}) + \lambda_2 \Big ( \sum_{i=1}^{K_u} \frac{t_{u,i}(2^{\frac{s_i}{\nu t_{u,i}B}} - 1)\Gamma_{1}\sigma_{1,i}^2}{N \gamma_i} \\
&+ \sum_{i=1}^{K_u}\kappa_i c_i (u_i - s_i) f_{u,i}^2\Big ) + \lambda_3 \left( \sum_{i=1}^{K_u} \frac{t_{d,i} (2^{\frac{\mu s_i}{t_{d,i} B}} - 1)\Gamma_{2}\sigma_{2,i}^2}{N \gamma_{i}} + \sum_{i = 1}^{K_u} \kappa_m d_m f_{mi}^2 s_i\right) + \sum_{i=1}^K \xi_i \Big(\frac{c_i q_i}{f_{u,i}} + t_{u,i} - T_{d} \Big) \\
&+\sum_{i=1}^K \beta_i (t_{u,i} - T_1)+ \sum_{i=1}^K \theta_i \Big(\frac{d_m s_i}{f_{m,i}} - T_2 \Big) + \sum_{i=1}^K \phi_i (t_{d,i} - T_3) + \lambda_5 \Big(\sum_{i=1}^{K} f_{m,i} - f_{m,\max}\Big) \tag{\theequation}
\end{align*}

Taking the derivative of the Lagrangian in (\ref{LDual}) with respect to $E_{\text{total}}$, $E_u$ and $E_m$ and setting it equal to zero results in $\lambda_0 = -1$, $\lambda_2 = 1 - w$ and $\lambda_3 = w$ respectively.

Applying Karush-Kuhn-Tucker (KKT) conditions
\subsubsection{with respect to offloading time $t_{u,i}$ in Phase I}
\setcounter{equation}{24}
\begin{align*}\label{diffT1}
\nabla_{t_{u,i}}&\mathcal{L} = (1 - w) \Bigg( \frac{\left(2^{\frac{s_i}{\nu t_{u,i} B}} - 1 \right) \Gamma_{1} \sigma_{1,i}^2}{N \gamma_i} - \frac{s_i \ln 2 \left(2^{\frac{s_i}{\nu t_{u,i} B}} \right) \Gamma_{1} \sigma_{1,i}^2}{\nu B t_{u,i} N \gamma_i} \Bigg) + \beta_i + \xi_i = 0\\
&\iff (1 - w) \left( f(x_{1,i}) - x_{1,i} f'(x_{1,i}) \right) + \beta_i + \xi_i = 0 \tag{\theequation}
\end{align*}
where $x_{1,i} = \frac{1}{t_{u,i}}$ and $f(x_{1,i}) = \frac{\left(2^{\frac{s_i}{\nu t_{u,i} B}} - 1 \right) \Gamma_{1} \sigma_{1,i}^2}{N \gamma_i}$. 
\subsubsection{with respect to downloading time $t_{d,i}$ in Phase III}
\setcounter{equation}{25}
\begin{align*}\label{diffT3}
\nabla_{t_{d,i}}\mathcal{L} &= w \Bigg( \frac{\left(2^{\frac{\mu s_i}{t_{d,i} B}} - 1 \right) \Gamma_{2} \sigma_{2,i}^2}{N \gamma_i} - \frac{\mu s_i \ln 2 \left(2^{\frac{\mu s_i}{t_{d,i} B}}\right) \Gamma_{2} \sigma_{2,i}^2}{B t_{d,i} N \gamma_i} \Bigg) + \phi_i = 0\\
& \iff w (f(x_{2,i}) - x_{2,i} f'(x_{2,i})) + \phi_i = 0 \tag{\theequation}
\end{align*}
where $x_{2,i} = \frac{1}{t_{d,i}}$ and $f(x_{2,i}) = \frac{\left(2^{\frac{\mu s_i}{t_{d,i} B}} - 1 \right) \Gamma_{2} \sigma_{2,i}^2}{N \gamma_i}$.

Substituting $y = -\frac{\beta_i + \xi_i}{(1 - w)}$, $x = x_{1,i} = \frac{1}{t_{u,i}}$, $c = \frac{\nu}{s_i}$, $\sigma^2 = \frac{ \Gamma_{1} \sigma_{1,i}^2}{N \gamma_i}$ in~(\ref{LambertSol}) for $t_{u,i}$, and $y = - \frac{\phi_i}{w}$, $x = x_{2,i} = \frac{1}{t_{d,i}}$, $c = \frac{1}{\mu s_i}$, $\sigma^2 = \frac{\Gamma_{2} \sigma_{2,i}^2}{N \gamma_i}$ in~(\ref{LambertSol}) for $t_{d,i}$, we get the uploading (downloading) times, $t_{u,i} (t_{d,i})$ in~(\ref{x1x2}-b), respectively.
\vspace{-1mm}
\subsection{Appendix C - Proof for Theorem 2}
Applying KKT conditions 
\subsubsection{with respect to the local CPU frequency at the $i^{\text{th}}$ user $f_{u,i}$}
\setcounter{equation}{26}
\begin{align*}\label{difffiapp}
\nabla_{f_{u,i}}\mathcal{L} &=  2 (1 - w) \kappa_i c_i (u_i - s_i) f_{u,i} - \xi_i \frac{c_i q_i}{f_{u,i}^2} = 0\\
&\iff  2 (1 - w) \kappa_i c_i (u_i - s_i) f_{u,i}^3 - \xi_i c_i (u_i - s_i) = 0  \tag{\theequation}
\end{align*}\setcounter{equation}{27}

\subsubsection{with respect to the MEC CPU frequency for computation of the $i^{\text{th}}$ user's task $f_{m,i}$}
\squeezeup \squeezeup
\begin{align*}\label{difffmiapp}
    \nabla_{f_{m,i}} \mathcal{L} &=  2 w \kappa_m d_m s_i f_{m,i} - \frac{\theta_i d_m s_i}{f_{m,i}^2} + \lambda_5 = 0 \tag{\theequation}
\end{align*}
Simplifying (\ref{difffiapp}) and (\ref{difffmiapp}) leads to (\ref{difffi}) and (\ref{difffmi}) respectively. The equation for $f_{mi}$ is in terms of the variable $s_i$ but can be solved in closed form as a root for the cubic equation of the form $af_{mi}^3 + bf_{mi}^2 + cf_{mi} + d$, where $a = 2W\kappa_m d_m s_i$, $b = \lambda_5$, $c = 0$ and $d = -\theta_i d_m s_i$.  

\bibliographystyle{./IEEEtran}
\bibliography{./newbib}
\end{document}